\documentclass[12pt,reqno]{amsart}

\usepackage{latexsym,amsmath}
\usepackage{verbatim}
\usepackage{tikz}
\usepackage{array}
\usepackage{amssymb}
\usepackage[pdftex]{epsfig}

\newcommand{\wep}{{\sl{w.e.p.}}}
\newcommand{\aas}{{\sl{a.a.s.}}}

\begin{document}
\newtheorem{thm}{Theorem}[section]
\newtheorem{lemma}[thm]{Lemma}
\newtheorem{definition}[thm]{Definition}
\newtheorem{conjecture}[thm]{Conjecture}
\newtheorem{proposition}[thm]{Proposition}
\newtheorem{algorithm}[thm]{Algorithm}
\newtheorem{corollary}[thm]{Corollary}
\newcommand{\noin}{\noindent}
\newcommand{\ind}{\indent}
\newcommand{\om}{\omega}
\newcommand{\pp}{\mathcal P}
\newcommand{\ppp}{\mathfrak P}
\newcommand{\N}{{\mathbb N}}
\newcommand\eps{\varepsilon}
\newcommand{\E}{\mathbb E}
\newcommand{\Prob}{\mathbb{P}}
\newcommand{\Rrr}{{\mathbb R}}
\newcommand{\sbrac}[1]{\left({\scriptstyle #1}\right)}
\newcommand{\tbrac}[1]{\left({\textstyle #1}\right)}
\newcommand{\smorl}[1]{{\scriptstyle #1}}
\newcommand{\sbfrac}[2]{\left(\frac{\scriptstyle #1}{\scriptstyle #2}\right)}
\newcommand{\tbfrac}[2]{\left(\frac{\textstyle #1}{\textstyle #2}\right)}
\newcommand{\sfrac}[2]{\frac{\scriptstyle #1}{\scriptstyle #2}}
\newcommand{\bfrac}[2]{\left(\frac{#1}{#2}\right)}
\newcommand{\brac}[1]{\left(#1\right)}
\newcommand{\Ppp}{{\mathbb P}}
\newcommand{\Eee}{{\mathbb E}}
\newcommand{\ul}[1]{\mbox{\boldmath$#1$}}
\newcommand{\ulu}[1]{\mbox{\bf #1}}
\def\a{\alpha} \def\b{\beta} \def\d{\delta} \def\D{\Delta}
\def\e{\epsilon} \def\f{\phi} \def\F{{\Phi}} \def\g{\gamma}
\def\G{\Gamma} \def\k{\kappa}\def\K{\Kappa}
\def\z{\zeta} \def\th{\theta} \def\TH{\Theta} \def\Th{\Theta} \def\la{\lambda}
\def\La{\Lambda} \def\m{\mu} \def\n{\nu} \def\p{\pi}
\def\r{\rho} \def\R{\Rho} \def\s{\sigma} \def\S{\Sigma}
\def\t{\tau} \def\om{\omega} \def\OM{\Omega} \def\Om{\Omega}

\title[Geometric Graph Properties of the SPA model]{Geometric Graph Properties of the Spatial Preferred Attachment model}

\author{Jeannette Janssen}
\address{Department of Mathematics and Statistics, Dalhousie University, PO Box 15000, Halifax, NS, Canada, B3H 4R2}
\curraddr{}
\email{janssen@mathstat.dal.ca}
\thanks{}

\author{Pawe\l{} Pra\l{}at}
\address{Department of Mathematics, Ryerson University, Toronto, ON, Canada, M5B 2K3}
\curraddr{}
\email{pralat@ryerson.ca}

\author{Rory Wilson}
\address{Department of Mathematics and Statistics, Dalhousie University, PO Box 15000, Halifax, NS, Canada, B3H 4R2}
\curraddr{}
\email{rwilson@mathstat.dal.ca}
\thanks{The authors gratefully acknowledge support from NSERC and MITACS grants.}

\date{}

\begin{abstract}
The spatial preferred attachment (SPA) model is a model for networked information spaces such as domains of the World Wide Web,  citation graphs, and on-line social networks. It uses a metric space to  model the hidden attributes of the vertices. Thus, vertices are elements of a metric space, and link formation depends on the metric distance between vertices. We show, through theoretical analysis and simulation, that for graphs formed according to the SPA model it is possible to infer the metric distance between vertices from the link structure of the graph. Precisely, the estimate is based on the number of common neighbours of a pair of vertices, a measure known as {\sl co-citation}. To be able to calculate this estimate, we derive a precise relation between the number of common neighbours and metric distance.  We also analyze the distribution of {\sl edge lengths}, where the length of an edge is the metric distance between its end points. We show that this distribution has three different regimes, and that the tail of this distribution follows a power law.
\end{abstract}


\maketitle

\section{Introduction}
Thanks to the World Wide Web and its hyperlinked structure, more and more information is becoming available in the form of a networked information space: a collection of information entities (documents, scientific papers, Web pages, individuals in a social network), connected by links between pairs of entities (references, citations, hyperlinks, ``friend'' relationships). Studies of various networked information spaces  have given convincing evidence that a significant amount of information about the entities represented by the vertex can be derived from the graph representing the link structure. This has led to the application of graph-theoretical techniques to such graphs, with the aim of developing methods to understand the link structure and mine its information. 

An important step in understanding the link structure is the development of a graph model; a stochastic process that models the link formation. The first generation of graph models  was mainly aimed at explaining the graph-theoretical properties observed in real-life networks. In such models, vertices are considered anonymous, and link formation is only influenced by the current link structure. An example is the seminal model by Barab\'{a}si and Albert in~\cite{ba} based on the principle of {\sl preferential attachment}: each new vertex attaches randomly to a prescribed number of existing vertices, with a link probability proportional to the degree, so vertices of high degree are more likely to receive a link from the new vertex. 

In networked information spaces, vertices are not only defined by their link environment, but also by the information entity they represent. More recently, attempts have been made to model this alternative view of the vertices through {\sl spatial models}. In a spatial model, vertices are embedded in a metric space, and link formation is influenced by the metric distance between vertices. The metric space is meant to be like a feature space, so that the coordinates of a vertex in this space represent the information associated with the vertex. For example, in text mining, documents are commonly represented as vectors in a word space.  The metric is chosen so that metric distance represents similarity, i.e.~vertices whose information entities are closely related will be at a short distance from each other in the metric space. 

In this paper, we focus on the Spatial Preferred Attachment (SPA) model, proposed in~\cite{spa1}, and analyze the relationship between the link structure of graphs produced by the model, and the relative positions of the vertices in the metric space. The SPA model generates directed graphs according to the following principle. Vertices are points in a given metric space. Each vertex $v$ has a {\sl sphere of influence}. The volume of the sphere of influence of a vertex is a function of its in-degree. A new vertex $u$ can only link to an existing vertex $v$ if $u$ falls inside the sphere of influence of $v$. In the latter case, $u$ links to $v$ with probability $p$. The SPA model incorporates the principle of preferential attachment, since vertices with a higher in-degree will have a larger sphere of influence.  A model for on-line social networks based on similar principles can be found in~\cite{geo-protean, wosn}.

A number of spatial models have been proposed recently~\cite{Bradonjic09,FFV06,FFV08,Masuda05,HRPrzulj08}.  In these models, as in the SPA model, the relationship between spatial distance and link formation is determined by a threshold function: a link is possible if vertices are within a prescribed threshold distance of each other, and impossible otherwise. However, for these models the threshold distance remains constant throughout the process, and does not depend on the degree, and decrease with time, as in the SPA model. 

A different class of graphs explores the interplay between distance and edge likelihood---with associated graph properties---with more involved mechanisms than simple thresholds: for example, in~\cite{esker}, each new vertex is born with $m$ edges, each joining a neighbour with probability proportional to the in-degrees and a function of the distance between them. Variations include the deterministic model~\cite{oliva} in which edges are formed based on the ``utility'' for the nodes in question, utility incorporating both in-degree and distance; in~\cite{ferretti}, the model demands that the number of nodes per unit volume is constant, and an analysis on the distribution of edge lengths is also included.  Beyond the creation of models,~\cite{krioukov} takes a closer look at the concept of complex networks having an underlying geometry. For a recent survey of spatial models, see~\cite{Janssen10}.

Our first main result shows that, for the SPA model, the number of common in-neighbours between a pair of vertices can, in many cases, be used to estimate the distance between the vertices. Since the metric distance is assumed to represent the similarity or ``closeness'' of the entities represented by the vertices, this means that it is possible to estimate similarity between vertices by looking at the graph only, i.e.~without considering the underlying reality represented by the metric space. The number of common in-neighbours in a citation graph is known in library science as the measure of {\sl  co-citation}, and is one  the earliest measures of graph-based similarity, proposed by Small in 1973 in~\cite{Small73}.  Co-citation, and the related measure of bibliographic coupling (from~\cite{Kessler63}) based on the number of common out-neighbours, are widely used link similarity measures for scientific papers, via the citation graph, for Web pages, and others~\cite{BichtEat80,DeanHenzinger99,Menczer04,LaiWu05}. 

The question of determining similarity between vertices is one that is central to many link mining applications. It is an important tool in searching, by finding documents or Web pages that are similar to a given target document. It can also be used as the basis to identify {\sl  communities}, or clusters, of similar vertices. A purely graph-based measure of similarity  can be used as a complementary indication of similarity between vertices when other information is unreliable (as is often the case in the World Wide Web), largely unavailable (as in some biological networks and online social networks), or protected by privacy laws (as in networks representing phone calls or bank transactions). 

Our result on the relationship between number of common neighbours and metric distance is derived theoretically through an analysis of the SPA model. The analytic result is asymptotic in the size of the graph. In order to test the result on realistic graph sizes, we performed simulations for graphs of 100,000  vertices, with various parameter choices. The simulations show that the real distance and the predicted distance from the number of common neighbours are in very good agreement.  

Our second main result determines the distribution of the {\sl edge lengths}, where the length of an edge is the metric distance between its end points. Edge length is a metric property of a graph feature, and edge length distribution is a combined metric/graph property which is unique to spatial graph models.  In the SPA model, the maximum length of an edge is determined by the size of the sphere of influence of its destination vertex, and this size is determined by the degree of the vertex. Since the degrees follow a power law, we might expect that the edge length distribution follows a power law. We show, both through theoretical results and simulations, that the situation is slightly more complex. In fact, we present clear evidence that, for a certain combination of model parameters, there are three different regimes of the distribution. For the smallest edge lengths, the cumulative edge length distribution is constant: almost all edges fall in this category. In the mid range, we have a power law with coefficient between 0 and 1, and in the tail, we have a power law with exponent greater than 1. 

In Section~2, we describe the SPA model and derive some properties on the degree of a vertex which we will need to establish our results. In Section~3, we give the result on common in-neighbours and metric distance, and present the simulation results. In Section~4, we state our theorem on edge length distribution, and present the edge length distribution as obtained through simulations for various parameters. In Section~5, we give proofs of all the main theorems. 
 
\section{The SPA model}

We start by giving a precise description of the SPA model, and deriving some facts about the degrees of the vertices, which we will need to prove our main results. In~\cite{spa1}, the model is defined for a variety of metric spaces $S$. In this paper, we let $S$ be the unit hypercube in $\Rrr^m$, equipped with the torus metric derived from any of the $L_p$ norms. This means that for any two points $x$ and $y$ in $S$,
\[
d(x,y)=\min \big\{ ||x-y+u||_p\,:\,u\in \{-1,0,1\}^m \big\}.
\] 
The torus metric thus ``wraps around'' the boundaries of the  unit square; this metric was chosen to eliminate boundary effects. Let $c_m$ be the constant of proportionality of volume used with the $m$-th power of the radius in $m$ dimensions, so the volume of a ball of radius $r$ in $m$-dimensional space with the given metric equals $c_m r^m$. For example, for the Euclidean metric, $c_2=\pi$, and for the product metric derived from $L_\infty$, $c_m=2^{m}$. 

The parameters of the model consist of the \emph{link probability} $p\in[0,1]$, and two positive constants $A_1$ and $A_2$. The SPA model generates stochastic sequences of graphs $ (G_{t}:t\geq 0)$, where $G_{t}=(V_{t},E_{t})$, and $V_{t}\subseteq S$. Let $\deg^{-}(v,t)$ be the in-degree of vertex $v$ in $ G_{t}$, and $\deg^+(v,t)$ its out-degree. We define the \emph{sphere of influence} $S(v,t)$ of vertex $v$ at time $t\geq 1$ to be the ball centered at $v$ with volume $|S(v,t)|$ defined as follows:
\begin{equation}\label{above}
|S(v,t)|=\frac{A_1{\deg}^{-}(v,t)+A_2}{t},
\end{equation}
or $S(v,t)=S$ and $|S(v,t)|=1$ if the right-hand-side of (\ref{above}) is greater than 1. 

The process begins at $t=0$, with $G_0$ being the null graph. Time-step $t$, $t\geq 1$, is defined to be the transition between $G_{t-1}$ and $G_t$. At the beginning of each time-step $t$, a new vertex $v_t$ is chosen \emph{uniformly at random} from $S$, and added to $V_{t-1}$ to create $V_{t}$. Next, independently, for each vertex $u\in V_{t-1}$ such that $v_t \in S(u,t-1)$, a directed link $(v_{t},u)$ is created with probability $p$. Thus, the probability that a link $(v_t,u)$ is added in time-step $t$ equals $p\,|S(u,t-1)|$. 

We note that, to avoid the resulting graph becoming too dense, the parameters must be chosen so that $pA_1 < 1$, as explained in~\cite{spa1}.  In this paper, we assume that the parameters meet this condition. Also, the original model as presented in~\cite{spa1} has a third parameter, $A_3$, which is assumed to be zero here.  This causes no loss of generality, since all asymptotic results presented here are unaffected by $A_3$.

We now introduce some more definitions. In the rest of the paper, unless otherwise stated we will assume all asymptotics to refer to $n$ going to infinity, where $n$ is the end time of the growth process, and thus the final size of the graph. (As explained above, Theorem~\ref{thm:exp_degree} is an exception.) We say that an event holds \emph{asymptotically almost surely} (\aas) if the probability that it holds tends to one as $n$ goes to infinity. Similarly, we will use \emph{with extreme probability} (\wep) if the event holds  with probability at least $1-\exp(-\Theta(\log^2 n))$.  Thus, if we consider a polynomial number of events that each holds \wep , then \wep\ all events hold.

It was shown in~\cite{spa1} that the SPA model produces graphs with a power law degree distribution, with exponent $1+1/(pA_1)$. In~\cite{spa_cfp}, the (directed) diameter of the model was investigated. For the results of this paper, we need a precise expression for the expected in-degree of each vertex.

\begin{thm}\label{thm:exp_degree} 
Let $\omega=\omega(t)$ be any function tending to infinity together with $t$. The expected in-degree at time $t$ of a vertex $v_i$ born at time $i \ge \omega$ is given by
\begin{equation}\label{exp_degree}
\Eee(\deg^-(v_i,t)) = (1+o(1))\frac{A_2}{A_1} \left(\frac{t}{i}\right)^{pA_1} - \frac {A_2}{A_1}.
\end{equation}
\end{thm}
\begin{proof}
In order to simplify calculations, we make the following substitution:
\begin{equation}\label{eq:subst}
X(v_i,t) = \deg^-(v_i,t) + \frac{A_2}{A_1}.
\end{equation}
It follows immediately from the definition of the process that
$$
X(v_i,t+1) =
\begin{cases}
      X(v_i,t)+1, & \text{with probability $\frac{pA_1X(v_i,t)}{t}$}\\
      X(v_i,t), & \text{otherwise}.
   \end{cases}
$$
Therefore,
\begin{align*}
\mathbb{E}(X(v_i,t+1)~~|~~X(v_i,t)) &= (X(v_i,t)+1)\frac{pA_1X(v_i,t)}{t} + X(v_i,t)\left(1-\frac{pA_1X(v_i,t)}{t}\right)\\
&=X(v_i,t)\left(1+\frac{pA_1}{t}\right), 
\end{align*}
and so
$$
\mathbb{E}(X(v_i,t+1)) = \mathbb{E}(X(v_i,t)) \left( 1+\frac{pA_1}{t} \right).
$$
Since all vertices start with in-degree zero, $X(v_i,i) = \frac{A_2}{A_1}$. Since $i \ge \omega$, one can use this to get
\begin{eqnarray*}
\mathbb{E}(X(v_i,t)) &=& \frac{A_2}{A_1}\prod_{j=i}^{t-1} \left( 1+\frac{pA_1}{j} \right) \\
&=& (1+o(1)) \frac{A_2}{A_1} \exp\left(\sum_{j=i}^{t-1}\frac{pA_1}{j}\right)\\
&=& (1+o(1)) \frac{A_2}{A_1} \exp\left(pA_1 \log \left( \frac {t}{i} \right) \right) \\
&=& (1+o(1))\frac{A_2}{A_1} \left(\frac{t}{i}\right)^{pA_1},
\end{eqnarray*}
and the assertion follows from~(\ref{eq:subst}).
\end{proof}

Theorem~\ref{thm:exp_degree} states that the \emph{expected} in-degree of an individual vertex born at time $i$ is asymptotically equal to  $\frac{A_2}{A_1} \left(\frac{t}{i}\right)^{pA_1} - \frac {A_2}{A_1}$, with an error term of order  $o((t/i)^{pA_1})$. (The asymptotics assume that $t$ is going to infinity, and $i$ is a growing function of $t$.) However, the in-degree of an individual vertex is not concentrated around its expected value. This is due to variation happening shortly after birth; whether or not the vertex receives in-links in the first few time steps after its birth greatly affects the size of its sphere of influence throughout the process, and therefore its final in-degree. 

We can circumvent this difficulty by considering the final in-degree of the vertex, and infer the growth history of the in-degree from there. Namely, from the in-degree of the vertex at end time $n$, we can obtain sharp bounds on the in-degree of the vertex during most of the process. This is expressed in the following theorem. First, define a injective function $f:\Rrr\rightarrow \Rrr$ by  
\[
f(i)=\frac{A_2}{A_1} \left(\frac{n}{i}\right)^{pA_1},
\] 
so $f(i)$ is the expected in-degree, at time $n$, of a vertex born at time $i$ (up to a multiplicative factor of $(1+o(1))$). Thus $f^{-1}(k)$ is the birth time of a vertex of final in-degree $k$, had the in-degree of the vertex remained close to its expected value during its entire lifetime. Moreover, the (asymptotic) expected in-degree at time $t$ of a vertex born at time $i$ can be given as $(A_2/A_1) f(i)/f(t)$ (provided that $i=i(n)$ tends to infinity with $n$). Thus, if a vertex of final in-degree $k$ has in-degree growth close to its expected value, then 
$$
t=f^{-1} \left(\frac{A_2k}{A_1a} \right)
$$ 
will be the approximate time when that vertex has in-degree $a$. The precise statement and proof of this discussion follows below in Theorem~\ref{thm:deg}, the main result of this section.

\begin{thm}\label{thm:deg}
Let $\omega=\omega(n)$ be any function tending to infinity together with $n$. The following statement holds \aas\ for every vertex $v$ for which $\deg^-(v,n)=k=k(n) \ge \omega \log n$. Let $i=f^{-1}(k)$, and let 
$$
t_k=f^{-1}\left( \frac{A_2k}{A_1\omega \log n} \right).
$$
Then, for all values of $t$ such that $t_k \le t \le n$, 
\begin{equation}
\label{eqn:deg}
\deg^-( v,t)=(1+o(1))\frac{A_2}{A_1}\left(\frac{t}{i}\right)^{pA_1}=(1+o(1)) \frac{A_2}{A_1} \cdot \frac{k}{f(t)}=(1+o(1)) k \left(\frac{t}{n}\right)^{pA_1}.
\end{equation}
\end{thm}

The theorem implies that once a given vertex accumulates $\omega \log n$ in-neighbours, the rest of the process (until time-step $n$) can be predicted with high probability; in fact, \aas\ we get a concentration around the expected value. Let us mention that it seems that the $\omega$ factor is needed to get a concentration result. However, without this factor, the order of the in-degrees still can be predicted: once the vertex has $\log n$ in-neighbours, we can bound the in-degree of this vertex so that the ratio between upper and lower bounds would \aas\ be a constant.

In order to prove Theorem~\ref{thm:deg}, we need strong results on the concentration of the in-degree throughout the process. These results, and the proof of the theorem, are given in Section~\ref{proofs}.

\section{Number of common neighbours and spatial distance}

The principles of the SPA model make it plausible that vertices that are close together in space will have a relatively high number of common neighbours. Namely, if two vertices are close together, their spheres of influence will overlap during most of the process, and any new vertex falling in the intersection of both spheres has the potential to become a common neighbour. Thus, the number of common neighbours (co-citation) should lead to a reliable measure of closeness in the metric space. In this section, we will quantify the relation between spatial distance and number of common in-neighbours, and show how it can be used to estimate distance. 

The term ``common neighbour" here refers to common in-neighbours. Precisely, a vertex $w$ is a common neighbour of vertices $u$ and $v$ if there exist directed links from $w$ to $u$ and from $w$ to $v$. Note that in our model this can only occur if $w$ is younger than $u$ and $v$, and, at its birth, $w$ lies in the intersection of the spheres of influence of $u$ and $v$. We use $cn(u,v,t)$ to denote the number of common in-neighbours of $u$ and $v$ at time $t$. 

Theorem~\ref{thm:cn} distinguishes three cases. The division into cases is based on the trend, as shown in Theorem~\ref{thm:deg}, that spheres of influence tend to shrink over time.  Thus, once the spheres of influence of two vertices have become disjoint, and their boundaries have some distance between them, it is not likely that they will overlap at any time after that. The cases therefore are distinguished by how the spheres of influence of $u$ and $v$ overlap, and when or whether they become disjoint. Figure~\ref{fig:cases} gives a pictorial representation of the three cases.  Consider two vertices $u$ and $v$ so that $v$ has smaller in-degree at time $n$ than $u$. Thus, the sphere of influence of $v$ tends to be smaller than that of $u$, and the likely birth time of $u$ is before that of $v$. 

\begin{figure}
\begin{center}
    \begin{tabular}{ | m{1cm} || p{2cm}| p{2cm} | c | }
    \hline
    \textbf{Case} & \textbf{Near birth of $v$}  & \textbf{Near end of process} & \\ \hline \hline
    
    1 &

\begin{tikzpicture}[scale=0.4]
\draw[step=0.5cm,color=gray] (-3,-3) grid (2,3);
\draw [fill=gray!50] (-0.5,0) circle (1.2cm);
\fill (-0.5,0) circle (2pt) node[below left] {$u$};
\draw [fill=gray!30] (1.25,2.25) circle (0.6cm);
\fill (1.25,2.25) circle (2pt) node[below right] {$v$};
\end{tikzpicture}
		& 

\begin{tikzpicture}[scale=0.4]
\draw[step=0.50cm,color=gray] (-3,-3) grid (2,3);
\draw [fill=gray!50] (-0.5,0) circle (1cm);
\fill (-0.5,0) circle (2pt) node[below left] {$u$};
\draw [fill=gray!30] (1.25,2.25) circle (0.3cm);
\fill (1.25,2.25) circle (2pt) node[below right] {$v$};
\end{tikzpicture}
 &
Too far
		\\ \hline
    2 &
\begin{tikzpicture}[scale=0.4]
\draw[step=0.50cm,color=gray] (-3,-3) grid (2,3);
\draw [fill=gray!50] (-0.5,0) circle (2.5cm);
\fill (-0.5,0) circle (2pt) node[below left] {$u$};
\draw [fill=gray!30] (.4,.4) circle (0.6cm);
\fill (.4,.4) circle (2pt) node[below right] {$v$};
\end{tikzpicture}
    & 
\begin{tikzpicture}[scale=0.4]
\draw[step=0.50cm,color=gray] (-3,-3) grid (2,3);
\draw [fill=gray!50] (-0.5,0) circle (2.1cm);
\fill (-0.5,0) circle (2pt) node[below left] {$u$};
\draw [fill=gray!30] (.4,.4) circle (0.25cm);
\fill (.4,.4) circle (2pt) node[below right] {$v$};
\end{tikzpicture}
&
Too close
     \\ \hline
    3 & 
\begin{tikzpicture}[scale=0.4]
\draw[step=0.50cm,color=gray] (-3,-3) grid (2,3);
\draw [fill=gray!50] (-0.5,0) circle (2.3cm);
\fill (-0.5,0) circle (2pt) node[below left] {$u$};
\draw [fill=gray!30] (1,1) circle (0.4cm);
\fill (1,1) circle (2pt) node[below right] {$v$};
\end{tikzpicture}
    & 
\begin{tikzpicture}[scale=0.4]
\draw[step=0.50cm,color=gray] (-3,-3) grid (2,3);
\draw [fill=gray!50] (-0.5,0) circle (1.5cm);
\fill (-0.5,0) circle (2pt) node[below left] {$u$};
\draw [fill=gray!30] (1,1) circle (0.20cm);
\fill (1,1) circle (2pt) node[below right] {$v$};
\end{tikzpicture}
& Just right
 \\
    \hline \hline
    \end{tabular}
\end{center}
\caption{The three cases of Theorem~\ref{thm:cn}\label{fig:cases}}
\end{figure}

In Case~1, $u$ and $v$ are so far apart that their spheres of influence never overlap, except maybe for a negligible initial time period near their birth. In this case, no vertex can fall in the spheres of influence of both $u$ and $v$, and thus $u$ and $v$ will acquire no common neighbours after the initial time period. Thus, they will have negligibly few common neighbours. In this case again, accurate prediction of the spatial distance between $u$  and $v$ is not possible: if $u$ and $v$ have very few common neighbours, we can only give a lower bound on their distance. 

In Case~2, $u$ and $v$ are so close that the sphere of influence of $v$ is contained within the sphere of influence of $u$ for almost all of its existence. In this case, the number of common neighbours of $u$ and $v$ is a constant proportion of the degree of $v$, due to the fact that each new vertex linking to $v$ will automatically be within the sphere of influence of $u$, and thus can link to $u$ as well (and does so with  probability $p$.) This means that $u$ and $v$ are too close for accurate prediction: if $cn(u,v,n)$ and $\deg^{-}(v,n)$ differ by a factor close to $p$ we can only give an upper bound on the spatial distance between $u$ and $v$.

In Case~3, the sphere of influence of $v$ is contained in that of $u$ near the birth of $v$, but the spheres become disjoint before the end of the process. The moment at which the separation occurs can be determined fairly precisely, and depends heavily on the distance between $u$ and $v$. Thus, for this case we have a formula for the number of common neighbours which involves the distance between $u$ and $v$, and the in-degree of both $u$ and $v$ at the end of the process. Reversing the formula, we can obtain a reliable estimate for the distance between $u$ and $v$ from the observable graph properties $cn(u,v,n)$, $\deg^{-}(u)$ and $\deg^{-}(v)$.

\begin{thm}\label{thm:cn}
Let $\omega=\omega(n)$ be any function tending to infinity together with $n$. The following holds \aas\
Let  $v_k$ and $v_\ell$ be vertices such that 
$$
k=\deg(v_k,n) \geq \deg(v_\ell,n)=\ell \ge \omega^2 \log n
$$
in a graph generated by the SPA model. Let $d=d(v_k,v_\ell)$ be the distance between $v_k$ and $v_\ell$ in the metric space. Finally, let $T = f^{-1}(\ell / (\omega \log n))$. 
Then,
\begin{enumerate}
\item[Case 1.] If $d \ge \eps (\omega \log n / T)^{1/m}$ for some $\eps > 0$, then 
$$
cn(v_\ell,v_k,n)=O(\omega \log n).
$$

\smallskip
\item[Case 2.] If $k \ge (1+\eps) \ell$ for some $\eps>0$ and 
\begin{equation}\label{eq:cond_for_d}
d \le \left( \frac {A_1 k+A_2}{c_m n} \right)^{1/m} - \left( \frac {A_1 \ell + A_2}{c_m n} \right)^{1/m} = \Theta \left(\left( \frac {k}{n} \right)^{1/m} \right),
\end{equation}
then  
$$
cn(v_\ell,v_k,n) = (1+o(1))p \ell.
$$ 

If $k = (1+o(1)) \ell$ and $d \ll (k/n)^{1/m} = (1+o(1)) (\ell/n)^{1/m}$, then $cn(v_\ell,v_k,n) = (1+o(1))p \ell$ as well.

\smallskip
\item[Case 3.] If $k \ge (1+\eps) \ell$ for some $\eps>0$ and 
\begin{equation}\label{eq:cond_for_d2}
\left( \frac {A_1 k+A_2}{c_m n} \right)^{1/m} - \left( \frac {A_1 \ell + A_2}{c_m n} \right)^{1/m} < d \ll  (\omega \log n / T)^{1/m},
\end{equation}
then  
\begin{equation}
\label{eq:cn}
cn(v_\ell,v_k,n) = C i_k^{-\frac{(pA_1)^2}{1-pA_1}} i_\ell^{-pA_1} d^{-\frac{mpA_1}{1-pA_1}} \left( 1+O \left( \left( \frac{i_k}{i_\ell} \right)^{pA_1/m} \right) \right),
\end{equation}
where $i_k=f^{-1}(k)$ and $i_\ell=f^{-1}(\ell)$ and $C=p A_1^{-1} A_2^{\frac{1}{1-pA_1}} c_m^{-\frac{pA_1}{1-pA_1}}$.

If $k=(1+o(1)) \ell$ and $\eps (k/n)^{1/m} < d \ll  (\omega \log n / T)^{1/m}$ for some $\eps>0$, then
$$
cn(v_\ell,v_k,n) = \Theta \left( i_k^{-\frac{(pA_1)^2}{1-pA_1}} i_\ell^{-pA_1} d^{-\frac{mpA_1}{1-pA_1}} \right).
$$ 
\end{enumerate}
\end{thm}

The importance of the theorem is that~(\ref{eq:cn})  gives a relationship between the distance between the vertices, their number of common neighbours, and their degrees. Since the number of common neighbours and the degrees are observable from the graph, the equation allows us to obtain an estimate for the (spatial) distance between the vertices using only basic graph parameters.

We tested the predictive power of our theoretical results on data obtained from simulations. The data was obtained from a graph with  100,000 vertices. The graph was generated from points randomly distributed in the unit square in $\Rrr ^2$ according to the SPA model described in Section~2, with $n=100,000$ and $p=0.95$, and $A_1=A_2=1$. 

First of all, we show that a blind approach to using the co-citation measure (number of common neighbours) does not work.  In Figure~\ref{fig:raw} we plot spatial distance versus number of common neighbours without further processing.  No relation between the two is apparent. 

\begin{figure}[ht]
\centerline{\includegraphics[width=0.45\textwidth]{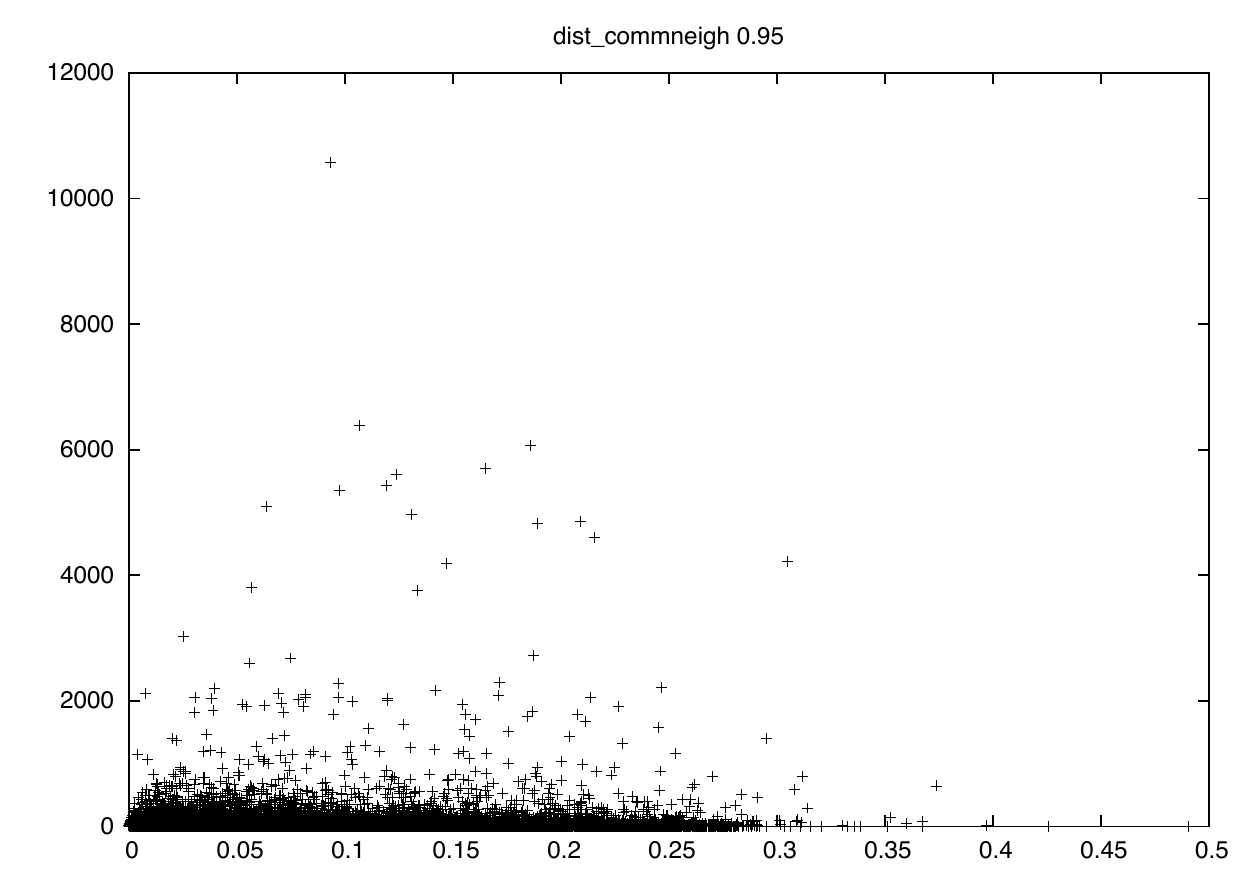}}
\caption{Actual distance vs.~number of common neighbours.\label{fig:raw}}
\end{figure}

Next, we apply  Theorem~\ref{thm:cn} to estimate the spatial distance between two vertices, based on the number of common neighbours of the pair.  (The spatial distance is actual distance between the point in the metric space, which for our simulation is the distance obtained from the Euclidean torus metric on  the unit square.) From Cases 1 and 2, we can only obtain a lower and upper bound on the distance, respectively. In order to eliminate Case 1 (too far), we consider only pairs that have at least 20 common neighbours. This reduces the data to  19,200 pairs. For pairs of vertices in Case 2 (too close), the number of common neighbours equals $p$ times the lowest degree of the pair. In order to eliminate this case,  we require that the number of common neighbours should be less than $p/2$ times the lowest degree of the pair. This reduces the data set to  2,400 pairs.  We expect these pairs mainly to be in Case 3. 

For pairs in Case 3, we can derive an estimate of the distance. Consider two such vertices $v_\ell$ and $v_k$, with final in-degree $\ell$ and $k$, respectively. We base our estimate on Equation~\ref{eq:cn}, where we ignore the multiplicative  $(1+O((\frac{i_k}{i_\ell})^{pA_1/m})$) error term. Namely, when $k$ and $\ell$ are of the same order, then this expression is the average of the lower and upper bound as derived in the proof of the theorem, and when  $\ell\ll k$ the term is asymptotically negligible. The estimated distance $\hat{d}$ between nodes $v_\ell$ and $v_k$, given that their number of common neighbours equals $N$, is then given by 

\begin{equation*}
\hat{d}=C'i_k^{-\frac{pA_1}{m}} i_\ell^{-\frac{1-pA_1}{m}} N^{-\frac{1-pA_1}{mpA1}} ,
\end{equation*}
where $i_k=f^{-1}(k)$ and $i_\ell=f^{-1}(\ell)$ and $C'=(p/A_1)^{\frac{1-pA_1}{mpA_1}} A_2^{\frac{1}{mpA_1}} c_m^{-\frac{1}{m}}$. 

Figure~\ref{fig:degree} shows actual vs.~estimated distance for these pairs.  The estimated distance (on the $y$-axis), is computed using only data obtainable from the graph: the in-degrees of both vertices, and their number of common neighbours. This is compared to the actual distance (on the $x$-axis), known from the simulation. We see almost perfect agreement between estimate and reality. 

\begin{figure}[ht]
\centerline{\includegraphics[width=0.4\textwidth]{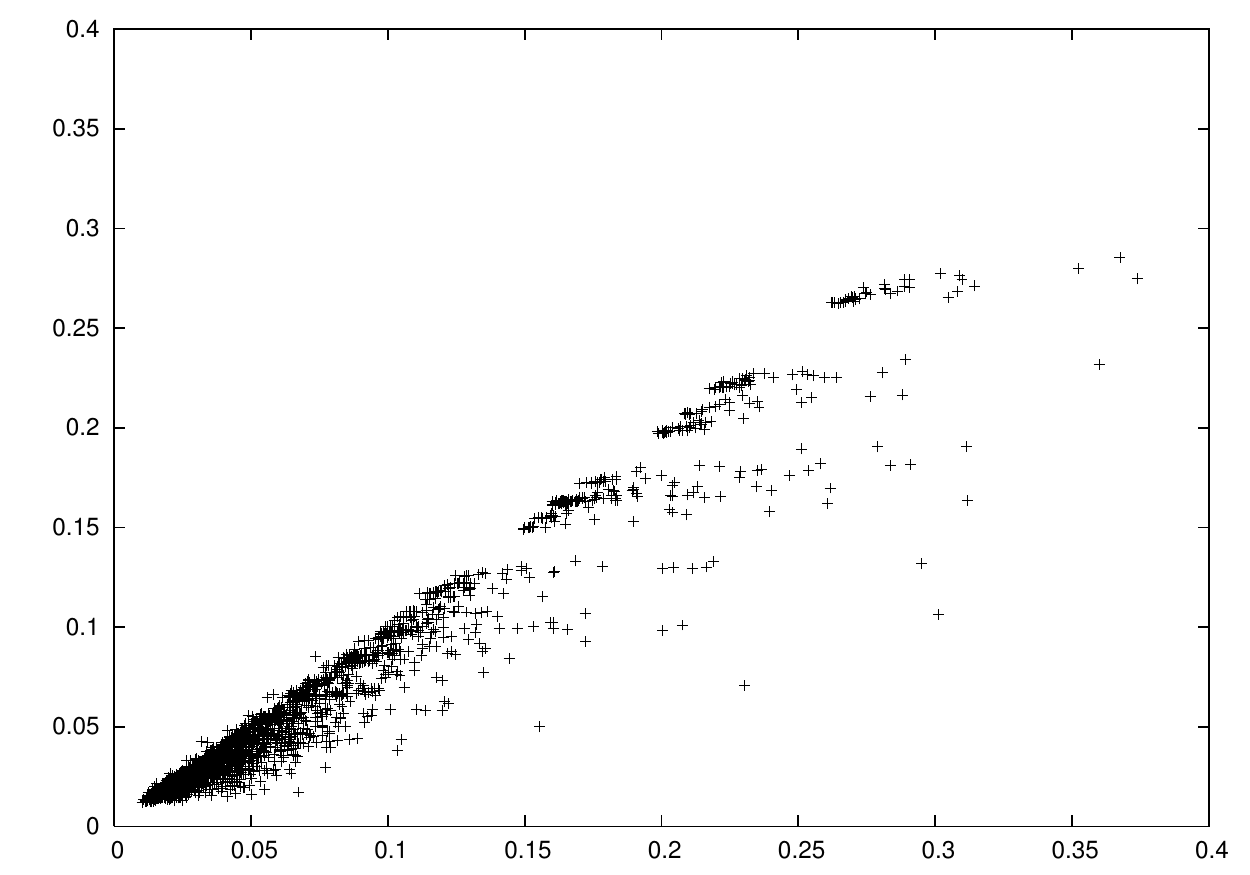}}
\caption{Actual distance ($x$-axis) vs.~estimated distance ($y$-axis) for eligible pairs from simulated data, calculated using the in-degree of both vertices.\label{fig:degree}}
\end{figure}

The figure shows that the scatter away from the diagonal is confined to points below the diagonal. This means that, for the corresponding pairs, the estimate $\hat{d}$ is lower than the actual distance. This is due to the choice to base our estimate on the average between the lower bound obtained from $t^-$, the estimated time when the sphere of influence of $v_\ell$ first touches the boundary of the sphere of influence of $v_k$, and the upper bound derived from $t^+$, when the spheres of influence of $v_\ell$ and $v_k$ first become disjoint. 

The probability that a neighbour of $v_\ell$ born between $t^-$ and $t^+$ becomes a common neighbour of $v_k$ and $v_\ell$ depends on the fraction of the sphere of influence of $v_\ell$ which lies inside the sphere of influence of $v_k$. If the curvature of the sphere of influence of $v_k$ is negligible so that the boundary locally resembles a line, and if the sphere of influence of $v_\ell$ remains constant in size from $t^-$ to $t^+$, then the average is a good estimate. However, both assumptions are notably false: the curvatures of the spheres of influence of $v_\ell$ an $v_k$ may well be of the same order, and the spheres of influence both shrink during the process.  This implies that the fraction of the sphere of influence of $v_\ell$ inside the sphere of influence of $v_k$ is smaller than assumed near time $t^+$, and larger than assumed near $t^-$.  Thus, the true expected number of common neighbours will likely be larger than indicated  by the average. This leads to an underestimate of the distance (more common neighbours is interpreted as closer distance).

In order to test our interpretation of the error in the estimation, we based the estimator $\hat{d}$ on a  convex combination of the lower bound $L$ on the numbers of common neighbours of  vertices $v_k$ and $v_\ell$ given by $L=p\deg^-(v_\ell,t^-)$ and the upper bound $U=p\deg(v_\ell,t^+)$. So the expected number of common neighbours is assumed to be $(1-c)L+cU$, which gives an expression involving $d$. Solving for $d$ gives our estimator $\hat{d}$. We found that the best value of $c$ occurred when $c=0.005$, which means that the lower bound based on time $t^-$ gives the best indication of the true number of common neighbours. 

The results for this adjusted estimator are given in Figure~\ref{fig:adjusted}. As we can see, the estimator is still not perfect; we conjecture that this is because the value of $c$ that gives the best estimate is not uniform over all pairs, but depends on the relative sizes of the spheres of influence of the pair in the critical time interval. 

\begin{figure}[ht]
\centerline{\includegraphics[width=0.4\textwidth]{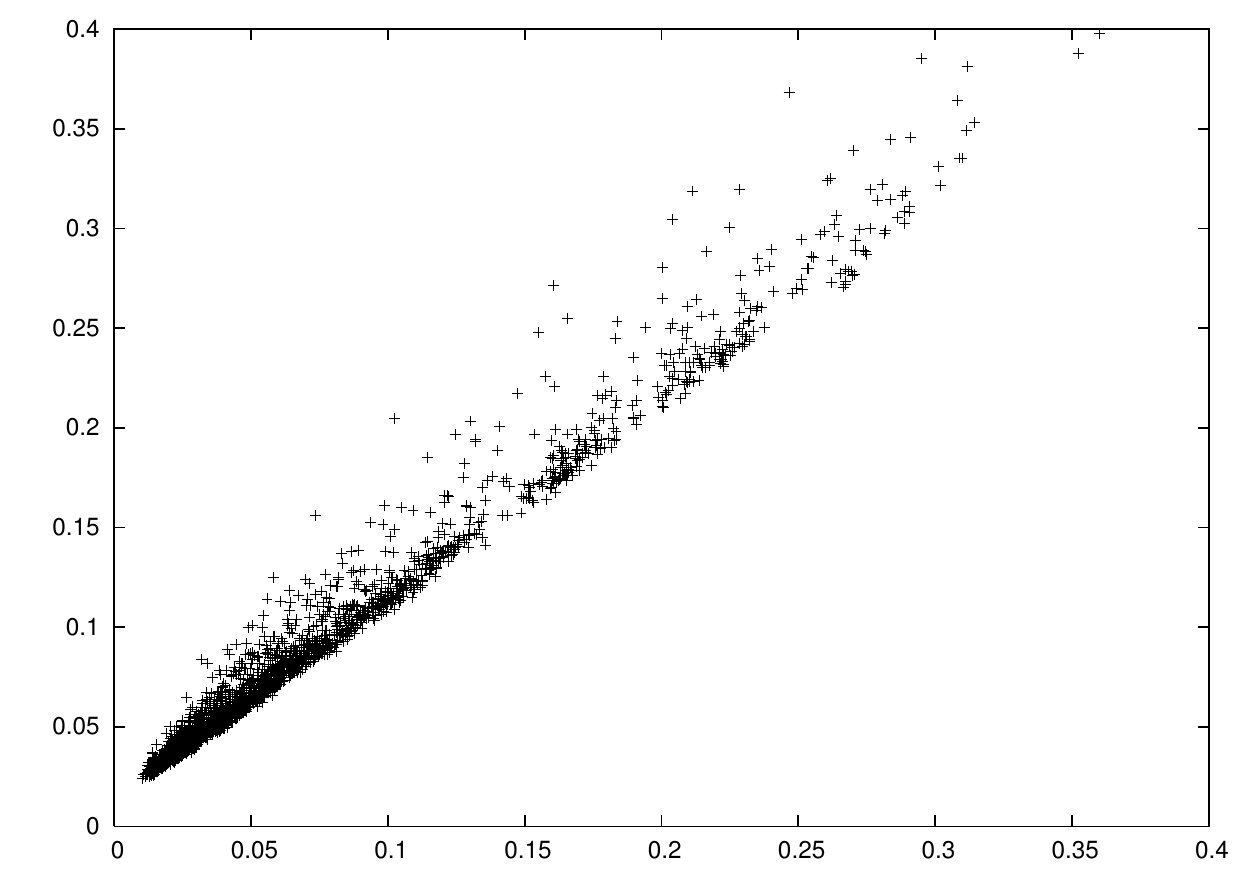}}
\caption{Actual distance ($x$-axis) vs.~estimated distance ($y$-axis) for eligible pairs from simulated data, using the adjusted estimator.\label{fig:adjusted}}
\end{figure}

\section{Edge length distribution}

In this section we derive the edge length distribution; that is, the number of edges whose length is at least a given value $x$. The length of an edge is the (metric) distance between its endpoints.  The edge distribution is a characteristic of spatial models. It will influence a number of graph properties, especially the diameter and the expansion properties. Long edges, even if they are rare, give the opportunity to jump to another locality in the metric space. It has been shown before (see, for example,~\cite{Kleinberg_Nature}) that a small number of long edges can reduce the average path length between vertices by a large factor.

In the SPA model, the degree distribution follows a power law, and the volume of the spheres of influence is proportional to the degree of a vertex. The radius of the sphere of influence determines the limit of the length of an edge to that vertex. Thus, we expect the edge lengths to follow a power law as well. These considerations lead us to consider all edges whose length exceeds a given value
\[
r_\alpha= \left( \frac {n^{-\alpha}}{c_m} \right)^{1/m}.
\] 
(Recall that $c_m$ is the volume of an $m$-dimensional ball of unit radius.) Namely, in this case we can limit our focus to those vertices whose sphere of influence has volume at least $n^{-\alpha}$.

Fix $\alpha >0$. An edge $(v,w) \in E(G)$ will be called a \emph{long edge} if the edge length $d(v,w) \ge r_{\alpha}$. We will study the random variable $e(\alpha)$, the number of long edges in the graph. Formally, 
$$
e(\alpha) = \left|\left\{(v,w) \in E: d(v,w) \geq  \left( \frac {n^{-\alpha}}{c_m} \right)^{1/m}\right\} \right|.
$$ 

\begin{thm}
\label{thm:edge}
In the SPA Model with $1/2<pA_1<1$, \aas\ the number of long edges is given by
\begin{equation}
e(\alpha) = 
   \begin{cases}
      (1+o(1)) \frac{pA_2}{1-pA_1} n, & \text{if $\alpha > 1$} \\
      (1+o(1)) C n^{2-\frac{1}{pA_1}+\alpha \frac{1-pA_1}{pA_1}}, & \text{if }1- \frac{pA_1}{4pA_1+2} < \alpha < 1,
   \end{cases} 
\end{equation}
where
$$
C = \frac {\Gamma \left(\frac{A_2}{A_1} +\frac{1}{pA_1} \right)} {\Gamma \left(\frac{A_2}{A_1} \right)} \frac{ A_1^{ \frac{pA_1}{1-pA_1}}}{1-pA_1} \left( \frac{(1-pA_1)^3}{2pA_1-1} A_1^{\frac{1-2pA_1}{(1-pA_1)pA_1}} + 1 - (pA_1)(1-pA_1) \right).
$$
\end{thm}

By~\cite{spa1}, the total number of edges in graphs generated by the SPA model equals $(1+o(1)) \frac{pA_2}{1-pA_1} n$. Thus, the first case of the theorem states that for $\alpha >1$, $e(\alpha)$ is approximately equal to the total number of edges. To see why this is so, consider that, as $\alpha$ increases, the threshold for an edge to be classified as ``long'', namely $r_\alpha$, decreases. If $\alpha >1$, then $r_\alpha$ is so small that almost all edges are long.  

The next range for $\alpha$,  $1- \frac{pA_1}{4pA_1+2} < \alpha < 1$, shows a linear relationship between $\log e(\alpha)$ and $\log r_\alpha $. Namely, $m\log r_\alpha = (1+o(1))(-\alpha) \log n$, and thus for this range,
\begin{eqnarray}
\label{eqn:midrange}
\log e(\alpha)  &=& (1+o(1)) \left(2-\frac{1}{pA_1}+\alpha \frac{1-pA_1}{pA_1} \right)\log n\\
&=&(1+o(1))\left((2-\frac{1}{pA_1})\log n -  m\frac{1-pA_1}{pA_1}\log  r_\alpha \right).\notag
\end{eqnarray}
Since $1/2<pA_1<1$, the slope of the line giving the relationship between $\alpha$ and $\log e(\alpha)$ lies between 0 and 1. 

The theorem does not include a claim about the tail of the edge distribution, when $\alpha$ becomes small, and thus $r_\alpha$ becomes relatively large. When $1-pA_1 < \alpha \le 1- \frac{pA_1}{4pA_1+2}$, the main contribution to $e(\alpha)$ comes from vertices that have very high final degree (not moderately high, as before) and the long edges are created till the very end of the process. Unfortunately, the number of vertices of very high degree cannot be precisely controlled ; from \cite{spa1} we only have upper bounds and lower bounds on the maximum degree that hold \wep\ which differ by a factor of $\log^4 n$. Therefore, it seems unlikely that $e(\alpha )$ is concentrated in this case. 

When $\alpha < 1-pA_1$, \aas\ long edges cannot be created at the end of the process but only until time $s = n^{\alpha/(1-pA_1)+o(1)}$. The main contribution to the number of long edges comes from those vertices that have very high degree at time $s$ (and have very high final degrees, of course). By a similar argument as given above, the number of such vertices, and thus the value of $e(\alpha )$, is not likely to be highly concentrated.

A different problem occurs when $pA_1<1/2$. The main contribution in this case comes from vertices born at time $\Theta(n^{\alpha})$ and the long edges must have been created when these vertices were still young, and had small degrees. Unfortunately, the behaviour of the random variable representing the degree of  a vertex is not concentrated until the degree is $\omega \log n$.  We expect $\Theta(n^{\alpha})$ such edges but we cannot control the behaviour of these vertices until the degree is large enough. 

The following theorem fills in the missing case when $\alpha $ is small. However, the results only apply to the {\sl expected} value of $e(\alpha)$, and they give broad results about the order of the exponent, instead of the finer results of the previous theorem. The proofs of both theorems can be found in the last section of the paper.

\begin{thm} \label{thm:edge-exp}
For the SPA model, the logarithmic behaviour of the expected value of $e(\alpha)$ is as follows. \\

For $1/2 < pA_1 < 1$,

\[
\frac{\log\E (e(\alpha))}{\log n} = 
   \begin{cases}
      1+o(1) & \text{if $\alpha \ge 1$}, \\
      {2-\frac{1}{pA_1}+\alpha \frac{1-pA_1}{pA_1}} +o(1), & \text{if }1- pA_1 < \alpha < 1, \\
      {\frac{\alpha p A_1}{1-pA_1}} +o(1), & \text{if } 0 \le \alpha \le 1- pA_1.
   \end{cases} 
\]

For $pA_1 < 1/2$, 
$$
\frac{\log\E (e(\alpha))}{\log n} = 
   \begin{cases}
      1+o(1) & \text{if $\alpha \ge 1$}, \\
      \alpha + o(1), & \text{if } 0 \le \alpha < 1.
   \end{cases} 
$$
\end{thm}

Thus, for the case where $pA_1>1/2$,  the middle range of $e(\alpha)$ extends beyond the lower bound on $\alpha$ for which precise results for $e(\alpha)$ can be obtained, and there is a third range for small $\alpha$, namely $\alpha <1-pA_1$, for which the expected relationship between $\log e(\alpha)$ and $\alpha $ is given by 
\begin{equation}
\label{eqn:end}
\log e(\alpha) =(1+o(1)) \left( \frac{ p A_1}{1-pA_1} \right) \alpha\log n = -\frac{m pA_1}{1-pA_1}\log r_\alpha.
\end{equation} 
Thus we have clear power law behaviour at the tail of the distribution, with coefficient $\frac{m pA_1}{1-pA_1}>1$. 

To verify our intuition that the real behaviour of the SPA model is similar to the asymptotic results given by the theorems, we ran  simulations. We generated graphs of 100,000 nodes, in $S$ of dimension $m=2$, for various values of $p$.  $A_1$ and $A_2$ were both set to $1$.  The results are seen in Figures~\ref{LEp08} and~\ref{LEp03}, where the logarithm of the number of long edges has been plotted against a range of values for $\alpha$. The straight lines in the figures represent the expected behaviour for the three ranges of $\alpha$ as given by~(\ref{eqn:midrange}) and~(\ref{eqn:end}) (and a horizontal line for the behaviour for large $\alpha$). To show the fact that the number of long edges decreases as the threshold $r_\alpha$ increases, the $x$-axis gives the values of $-\alpha$.

\begin{figure} [h] 
\begin{center}

\includegraphics[width=0.45\textwidth]{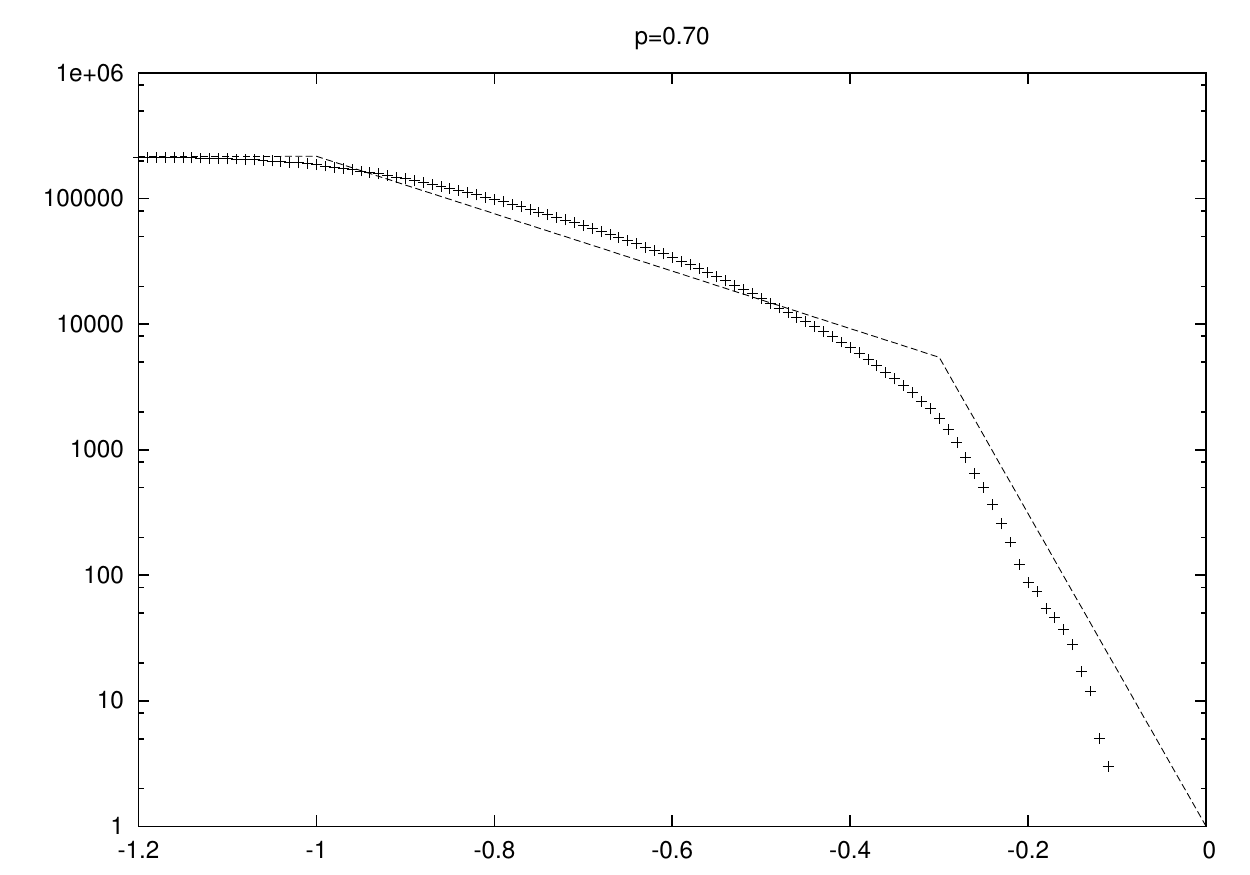}
\includegraphics[width=0.45\textwidth]{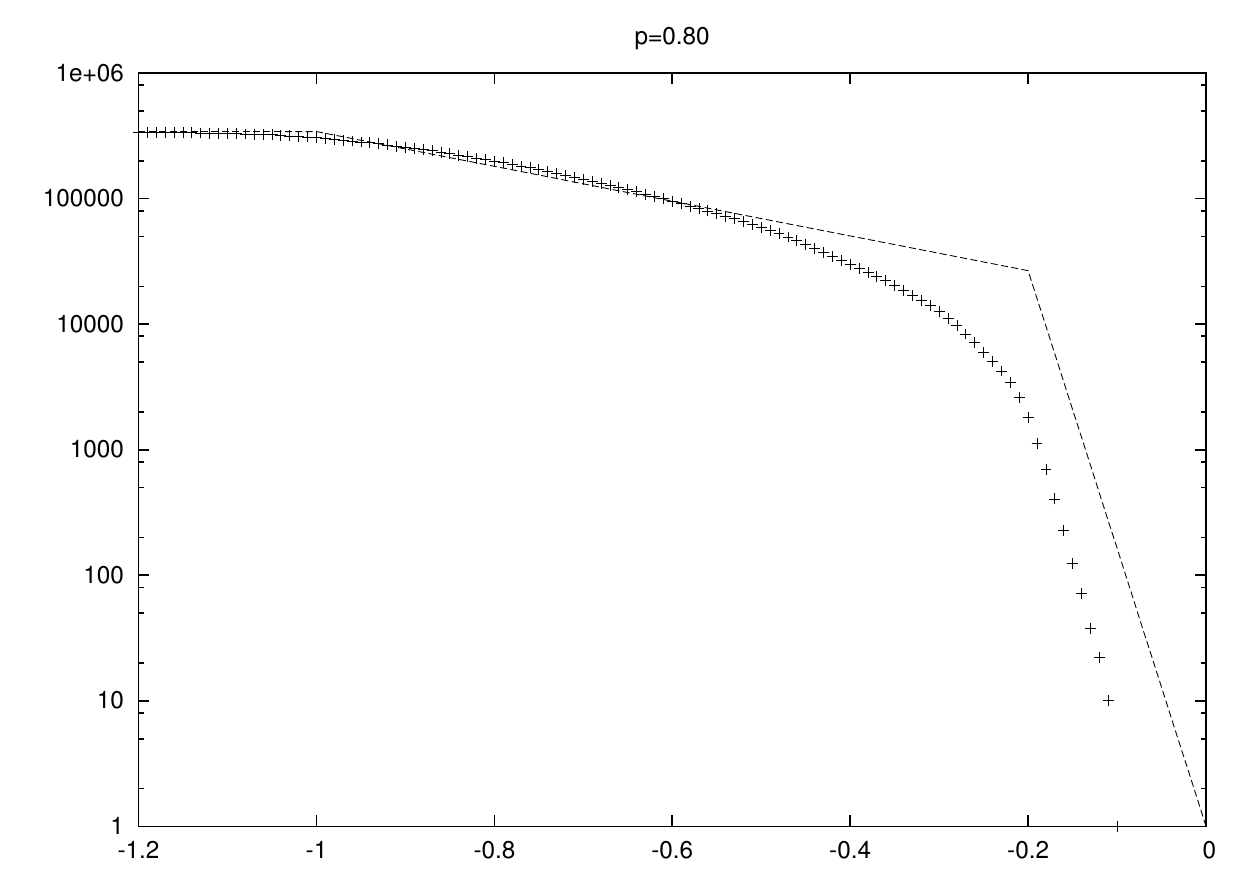}
\caption[Long Edges Simulation vs.\ Theory, $p=0.8$]{\small Long Edges Simulation vs.\ Theory, SPA Model. Parameters: $n=100,000$, $A_1=A_2=1$, $m=2$, $p=0.7$ (left) and $p=0.8$ (right).\label{LEp08}}
\end{center}
\end{figure}

Figure~\ref{LEp08} shows two values in the range $1/2<pA_1<1$. For both cases, the theoretical results expressed in Equations~\ref{eqn:midrange} and~\ref{eqn:end} give a good approximation of the envelope of the curve represented by the simulated values. Not surprisingly, near the threshold $1-pA_1=\alpha$, the simulated version shows smooth behaviour that is a blend between the behaviour on both sides of the range. The angle of the tail of the distribution has good agreement with the value predicted from the modified model. 

\begin{figure} [h] 
\begin{center}
\includegraphics[width=0.45\textwidth]{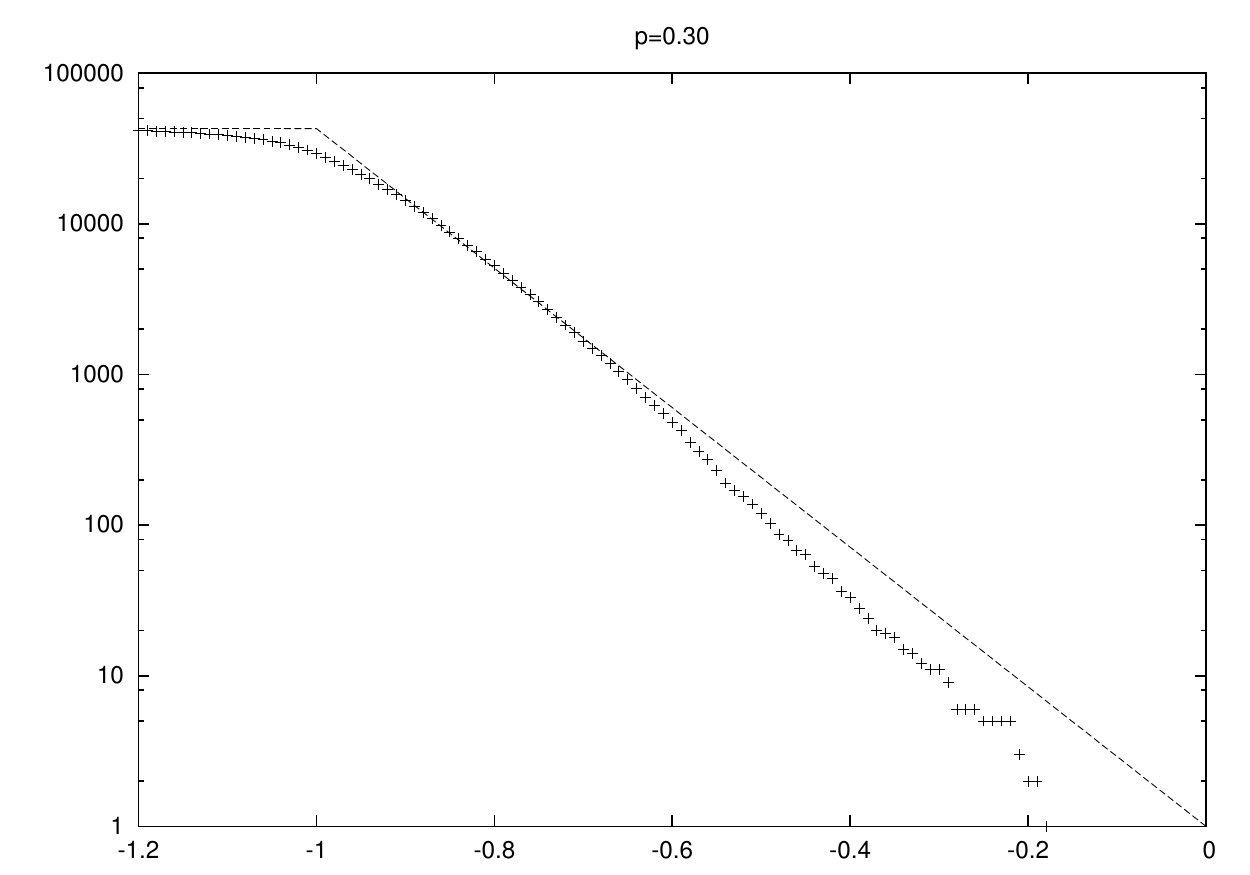}
\includegraphics[width=0.45\textwidth]{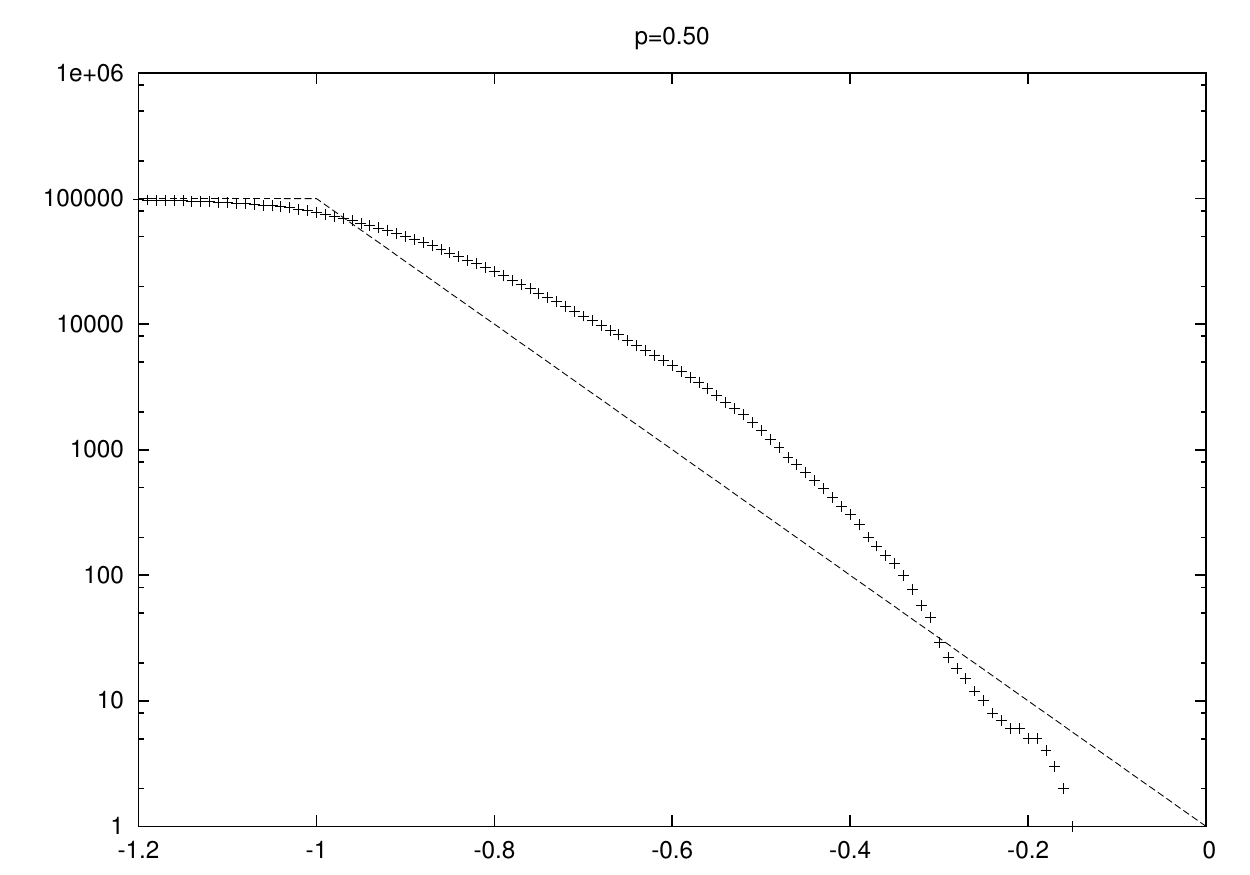}
\caption[Long Edges Simulation vs.\ Theory, $p=0.5$]{\small Long Edges Simulation vs.\ Theory, SPA Model. Parameters: $n=100,000$, $A_1=A_2=1$, $m=2$, $p=0.3$ (left), and $p=0.5$ (right)\label{LEp03}}
\end{center}
\end{figure}

In Figure \ref{LEp03} we give a simulation result for the case where $pA_1<1/2$. Here the modified model predicts only two regimes, which is borne out by the simulation data. We also include a picture for the case $pA_1=1/2$. At this cross-over value, no linear relationship between $\log e(\alpha)$ and $\alpha$ can be observed from the picture. However, our theoretical results predict that for larger values of $n$, the curve should approach a straight line with slope $-\alpha$.

\section{Proofs of the main theorems}\label{proofs}

\subsection{Degree of a vertex}
The first part  of this section is devoted the proof of Theorem~\ref{thm:deg}. We will be using the following version of a well-known Bernstein inequalities many times so let us state it explicitly.
\begin{lemma}[\cite{JLR}]\label{lem:Chernoff} 
Let $X$ be a random variable that can be expressed as a sum $X=\sum_{i=1}^{n} X_{i}$ of independent random indicator variables where $X_{i}\in\mathrm{Be}(p_{i})$ with (possibly) different $p_{i}=\mathbb{P} (X_{i} = 1)= \mathbb{E}X_{i}$. Then the following holds for $t \ge0$:
\begin{align*}
\mathbb{P} (X \ge\mathbb{E }X + t)  &  \le\exp\left(  - \frac{t^{2}}{2(\mathbb{E }X+t/3)} \right)  ,\\
\mathbb{P} (X \le\mathbb{E }X - t)  &  \le\exp\left(  - \frac{t^{2}}{2\mathbb{E }X} \right)  .
\end{align*}
In particular, if $\varepsilon\le3/2$, then
\begin{align}\label{eq:Ch}
\mathbb{P} (|X - \mathbb{E }X| \ge\varepsilon\mathbb{E }X)  &  \le2 \exp\left(  - \frac{\varepsilon^{2} \mathbb{E }X}{3} \right)  .
\end{align}
\end{lemma}

Now, we are ready to prove the following key observation.
\begin{thm}\label{thm:conc}
Suppose that $\deg^-(v,T)=d \ge \omega \log n$, where $\omega=\omega(n)$ is any function tending to infinity together with $n$. Then, for every value of $t$, $T \le t \le 2T$, we get that
$$
\left| \deg^-(v,t) - d \cdot \left( \frac tT \right)^{pA_1} \right| \le \frac {2}{pA_1} \cdot \frac {t}{T} \sqrt{d \log n}
$$
with probability $1-O(n^{-4/3})$. 
\end{thm}
\begin{proof}
Let $\omega=\omega(n)$ be any function tending to infinity together with $n$. Suppose that $\deg^-(v,T)=d \ge \omega \log n$. We will show that the upper bound holds; the lower bound can be obtained by using an analogous symmetric argument.

Let us introduce the following stopping time
$$
T_0=\min \left\{ t \ge T : \deg^-(v,t) >  d \cdot \left( \frac tT \right)^{pA_1} + \frac {2}{pA_1} \cdot \frac {t}{T} \sqrt{d \log n} ~~\vee~~ t=2T+1 \right\}.
$$ 

A \emph{stopping time} is any random variable $T_0$ with values in $\{0,1,\dots \} \cup \{\infty\}$ such that it can be determined whether $T_0 = t^*$ for any time $t^*$ from knowledge of the process up to and including time $t^*$. The name can be misleading, since a process does not \emph{stop} when it reaches a stopping time. Here, $T_0$ determines the first time the process does \emph{not} exhibit the bounded behaviour we wish to establish. The condition $t=2T+1$ has been added to assure that the set is never empty, and thus $T_0$ is well-defined. If $T_0=2T+1$, then the in-degree of $v$ remained bounded as given during the entire time interval $T\leq t\leq 2T$. In order to prove the bound, we need to show that with probability $1-O(n^{-4/3})$ we have $T_0=2T+1$. 

Suppose that $T_0 \le 2T$. Note that for $t \ge T$ up to and including time-step $T_0-1$, the random variable $\deg^-(v,t)$ is (deterministically) bounded from above, and so the number of new neighbours accumulated during this phase of the process, $\deg^-(v,T_0)-\deg^-(v,T)$, can be (stochastically) bounded from above by the sum $X = \sum_{t=T}^{T_0-1} X_t$ of independent indicator random variables $X_t$ with 
$$
\Prob(X_t = 1) = p \ \frac {A_1 \left( d \left( \frac tT \right)^{pA_1} + \frac {2}{pA_1} \cdot \frac {t}{T} \sqrt{d \log n} \right) + A_2}{t}.
$$
Hence,
\begin{eqnarray*}
\E \deg^-(v,T_0) &\le& d + \E X ~~=~~ d + \sum_{t=T}^{T_0-1} \E X_t \\
&=& d + p A_1 d T^{-p A_1} \left( \sum_{t=T}^{T_0-1} t^{pA_1-1} \right) + \frac {T_0-T}{T}  2 \sqrt{d \log n} + O(1) \\
&=& d \left( \frac {T_0}{T} \right)^{pA_1} + \frac {T_0-T}{T}  2 \sqrt{d \log n} + O(1).
\end{eqnarray*}
This implies that
\begin{eqnarray*}
\deg^-(v,T_0) - \E \deg^-(v,T_0) &\ge& \frac {2}{pA_1} \cdot \frac {T_0}{T} \sqrt{d \log n} - \frac {T_0-T}{T}  2 \sqrt{d \log n} - O(1) \\
&\ge& 2 \sqrt{d \log n},
\end{eqnarray*}
and it follows from the bound~(\ref{eq:Ch}) that
$$
\Prob (|X - \E X| \ge 2 \sqrt{d \log n}) \le 2 \exp \left(  - \ \eps \ \frac{2 \sqrt{d \log n}}{3} \right),
$$
where $\eps = 2 \sqrt{d \log n} / \E X$. Since the maximum value of $\E X$ corresponds to $T_0=2T$, it follows that $\E X \le d (2^{pA_1} - 1)(1+o(1)) \le d$, and so $\eps \ge 2 \sqrt{d^{-1} \log n}$. Therefore, the probability that $T_0 \le 2T$ is at most $2 \exp (- \frac 43 \log n)$ and the theorem is finished.
\end{proof}

Now, with Theorem~\ref{thm:conc} in hand we can easily get Theorem~\ref{thm:deg}. For a given vertex $v$ of degree $\omega \log n$ at time $T$ we obtain from Theorem~\ref{thm:conc} that, with probability $1-O(n^{-4/3})$, 
$$
d \left( \frac tT \right)^{pA_1} \left(1 - \frac {4}{pA_1} \sqrt{d^{-1} \log n} \right) \le \deg^-(v,t) \le d \left( \frac tT \right)^{pA_1} \left(1 + \frac {4}{pA_1} \sqrt{d^{-1} \log n} \right)
$$
for $T \le t \le 2T$. We can now keep applying the same theorem for times $2T$, $4T$, $8T$, $16T, \dots$, using the final value as the initial one for the next period, to get the statement for all values of $t$ from $T$ up to and including time $n$. Since we apply the theorem $O(\log n)$ times (for a given vertex $v$), the statement holds with probability $1-o(n^{-1})$ and so \aas\ the statement we are about to prove will hold for all vertices.

It remains to make sure that the accumulated multiplicative error is still only $(1+o(1))$. After applying the theorem recursively $i$ times the degree is shown to be $d 2^{pA_1i}(1+o(1))$. Using this rough estimate, and assuming the theorem is applied for a total of $k=O(\log n )$ times, we get that the error term is, in fact, bounded from above by
\begin{eqnarray*}
\prod_{i=1}^{k} \left( 1 + \frac{5}{pA_1} \sqrt{d^{-1} 2^{-pA_1i} \log n} \right) &=& (1+o(1)) \exp \left( \frac{5}{pA_1} \sqrt{d^{-1} \log n} \sum_{i=1}^{k} 2^{-pA_1i/2} \right) \\
&=& (1+o(1)) \exp \left( O(\sqrt{d^{-1} \log n}) \right) \\
&=& 1+o(1),
\end{eqnarray*}
since $d$ grows faster than $\log n$. A symmetric argument can be used to show a lower bound for the error term and so  Theorem~\ref{thm:deg} holds.

\subsection{Number of common neighbours}

The proof of Theorem~\ref{thm:cn}, which gives a formula for the number of common neighbours of two given vertices $v$ and $w$, is based on three cases, as explained in Section 3 and Figure~\ref{fig:cases}. 
The division into three cases is based on the trend, as shown in Theorem \ref{thm:deg}, that spheres of influence tend to shrink over time. It can happen that spheres of influence that are disjoint become overlapping at a later time instance, and thus do not fit any of the three cases. However, this behaviour happens with low enough probability that it does not affect our result.

\begin{proof}[Proof of Theorem~\ref{thm:cn}.]
The proof depends heavily on Theorem~\ref{thm:deg}. Any precise reference to the theorem will therefore be omitted. We can assume that at time $T$, 
$$
\deg(v_\ell, T) = (1+o(1)) \frac {A_2}{A_1} \omega \log n
$$ 
and the degree of this vertex is as predicted by (\ref{eqn:deg}) until the end of the process (that is, the ratio between the upper and lower bound on the degree is deterministically equal to $(1+o(1))$). Since $k \ge l$, the degree of $v_k$ for the time interval after $T$ is given by (\ref{eqn:deg}) as well. Let $r(v,t)$ denote the radius of the sphere of influence around $v$ at time $t$; that is, $r(v,t) = (|S(v,t)|/c_m)^{1/m}$.

\textbf{Case 1: } Suppose that $d \ge \eps (\omega \log n / T)^{1/m}$ for some $\eps > 0$. For $T \le t \le n$, we can deduce from the expression for the degree of $v_\ell$ over time and the expression for the volume $|S(v_\ell,t)|$ of the sphere of influence of $v_\ell$ that
$$
r(v_\ell, t) = (1+o(1)) \left( \frac {A_2 \omega \log n ( t / T )^{pA_1}}{c_m t} \right)^{1/m}.
$$
In particular, let us note that $d$ is of greater or equal order as $r(v_\ell,T)$, and hence of greater or equal order as $r(v_k,T)$ as well. Moreover, both radii tend to be decreasing from time $T$ on. (Formally what we mean is that $r(v_\ell,t) > r(v_\ell,t(1+\eps))$ for any $\eps>0$ and $t\ge T$. When a vertex receives a new neighbour, its radius slightly increases.) Therefore, there exists a constant $c=c(\eps)>0$ such that $S(v_\ell,t)$ and $S(v_k,t)$ are disconnected for every $t > cT$ and so there is no chance to create more common neighbours.  Since at time $cT$ the degree of vertex $v_\ell$ is $(1+o(1)) (A_2/A_1) c^{pA_1} \omega \log n = O(\omega \log n)$, we can apply an obvious upper bound to get
$$
cn(v_\ell,v_k,n) \le \deg(v_\ell,n) = O(\omega \log n).
$$ 
Finally, note that it can happen that $cT > n$, which means that  the process stops before the spheres of influence become disjoint. This causes no problem since the upper bound for the number of common neighbours at time $cT$ will then trivially hold at time $n$.

\textbf{Case 2: } Suppose $k \ge (1+\eps) \ell$ for some $\eps>0$ and $d$ satisfies inequality~(\ref{eq:cond_for_d}). Note that the condition for $d$ implies that at time $n$ the sphere of influence of $v_\ell$ is contained in that of $v_k$. Moreover, the radii of influence are proportionally decreasing during the process from the time we start having concentrated behaviour of degrees onwards (that is, from time $T$ on, in the sense explained earlier). So the sphere of influence of $v_\ell$ is contained in the sphere of influence of $v_k$ from  time $T$ to time $(1+o(1))n$. Any vertex $u$ that links to $v_\ell$ lies inside the sphere of influence of $v_\ell$ and thus of $v_k$ as well, and has a probability $p$ of also linking to $v_k$.  

At the end of the process (for $t=(1+o(1))n$) it can happen that the sphere of influence $S(v_\ell,t)$ is not completely contained in $S(v_k,t)$,  but it is the case that they overlap to a large extend, namely 
\begin{equation}\label{eq:ratio}
\frac {|S(v_\ell,t) \cap S(v_k,t)|}{|S(v_\ell,t)|} = 1+o(1).
\end{equation}
Thus, the probability that a neighbour of $v_\ell$, added during this phase of the process, is also a neighbour of $v_k$ is $(1+o(1))p$. 

Therefore, $\E cn(v_\ell,v_k,n) = (1+o(1)) p \ell$, since the number of common neighbours accumulated until time $T$ is $O(\omega \log n)$ and so is negligible. 

Suppose now that $k = (1+o(1)) \ell$ and $d \ll (k/n)^{1/m}$. In this case, the radii of $v_\ell$ and $v_k$ are approximately equal from time $T$ to the end of the process (that is, they differ by a multiplicative factor of $(1+o(1)))$. Since $d$ is of order smaller than the radii at the end of the process, property~(\ref{eq:ratio}) holds for $T \le t \le n$ and the results holds by the same argument as before.

\textbf{Case 3:}  Suppose $k \ge (1+\eps) \ell$ for some $\eps>0$ and $d$ satisfies inequality~(\ref{eq:cond_for_d2}). Note that the condition for $d$ implies that at time $T$ the sphere of influence of $v_\ell$ is contained in that of $v_k$, but this is not the case at time $n$.

Let $t^-$ be the first moment when $S(v_\ell,t)$ is not completely contained in $S(v_k,t)$ ($T < t^- \le n$). Let $t^+$ be the last time when the spheres overlap ($t^- \le t^+$). (Note that it is possible that $t^+ > n$ but, as before, this causes no problem.) Up to time $t^-$, each neighbour of $v_\ell$ will be a common neighbour of $v_\ell$ and $v_k$ with probability $p$. From time $t^+$ to $n$, no common neighbours can be created. From time $t^-$ until time $t^+$, the probability that a neighbour of $v_\ell$ becomes a neighbour of $v_k$ is \emph{at most} $p$.  Thus, $p \deg^-(v_\ell,t^-)$ and $p \deg^-(v_\ell,t^+)$ form a lower and an upper bound, respectively, on the expected number of common neighbours of $v$ and $w$.
 
Note that at time $t^-$, $S(v_\ell,t^-)$ is contained in $S(v_k,t^-)$ and ``touches" the boundary from the inside (the distance between the boundaries at time $t^{-}$ may not be exactly zero but certainly is $o(d)$). At time $t^+$, $S(v_\ell,t^+)$ is outside $S(v_k,t^-)$ but ``touches" the boundary from the outside. Since the centers of $S(v_\ell,t)$ and $S(v_k,t)$ are at distance $d$ from each other, this translates into the following expressions involving $t^-$ and $t^+$:
\begin{eqnarray*}
r(v_k, t^-) - r(v_\ell,t^-) &=& (1+o(1)) d,\\
r(v_k, t^+) + r(v_\ell,t^+) &=& (1+o(1)) d.
\end{eqnarray*}

Using the concentration result about the in-degree, this translates into the following conditions on $t^-$ and $t^+$
\begin{eqnarray*}
\left( \frac {A_2}{c_m} (t^-)^{pA_1-1} \right)^{1/m} i_k^{-pA_1/m} \left( 1 - \left( \frac {i_k}{i_\ell} \right)^{pA_1/m} \right) &=& (1+o(1)) d, \\
\left( \frac {A_2}{c_m} (t^+)^{pA_1-1} \right)^{1/m} i_k^{-pA_1/m} \left( 1 + \left( \frac {i_k}{i_\ell} \right)^{pA_1/m} \right) &=& (1+o(1)) d, \\
\end{eqnarray*}
and so
\begin{eqnarray*}
t^- &=& (1+o(1)) \left( \frac {A_2}{c_m} \right)^{ \frac{1}{1-pA_1}} i_k^{-\frac{pA_1}{1-pA_1}} d^{-\frac{m}{1-pA_1}}
             \left(1-\left( \frac{i_k}{i_\ell} \right)^{pA_1/m} \right)^{\frac{m}{1-pA_1}}, \\ 
t^+ &=& (1+o(1)) \left( \frac {A_2}{c_m} \right)^{ \frac{1}{1-pA_1}} i_k^{-\frac{pA_1}{1-pA_1}} d^{-\frac{m}{1-pA_1}}
             \left(1+\left( \frac{i_k}{i_\ell} \right)^{pA_1/m} \right)^{\frac{m}{1-pA_1}}. 
\end{eqnarray*}

The number of common neighbours of $v_k$ and $v_\ell$ is bounded from below by $(1+o(1))p\deg^-(v_\ell,t^-)$, and from above by $(1+o(1))p\deg(v_\ell,t^+)$. Using our knowledge about the behaviour of the in-degree of $v_\ell$, this leads to the  following bounds, which hold within a $(1+o(1))$ term: 
\begin{eqnarray*}
&p A_1^{-1} A_2^{\frac{1}{1-pA_1}} c_m^{-\frac{pA_1}{1-pA_1}} i_k^{-\frac{(pA_1)^2}{1-pA_1}} i_\ell^{-pA_1} d^{-\frac{mpA_1}{1-pA_1}} 
\left(1-\left( \frac{i_k}{i_\ell} \right)^{pA_1/m} \right)^{\frac{mpA_1}{1-pA_1}} \\
&\leq\Eee\, cn(v_\ell,v_k,n) \leq \\
&p A_1^{-1} A_2^{\frac{1}{1-pA_1}} c_m^{-\frac{pA_1}{1-pA_1}} i_k^{-\frac{(pA_1)^2}{1-pA_1}} i_\ell^{-pA_1} d^{-\frac{mpA_1}{1-pA_1}} 
\left(1+\left( \frac{i_k}{i_\ell} \right)^{pA_1/m} \right)^{\frac{mpA_1}{1-pA_1}}.
\end{eqnarray*}
The result follows from the fact that
\[
\left(1 \pm \left( \frac{i_k}{i_\ell} \right)^{pA_1/m} \right)^{\frac{mpA_1}{1-pA_1}}
=1+O \left( \left( \frac{i_k}{i_\ell} \right)^{pA_1/m} \right). 
\]

Finally, consider the case where $k=(1+o(1))\ell$, and thus $i_k/i_l=1+o(1)$. As before, from time 
$$
t^+ = (1+o(1)) \left( \frac {A_2 2^m}{c_m} \right)^{ \frac{1}{1-pA_1}} i_k^{-\frac{pA_1}{1-pA_1}} d^{-\frac{m}{1-pA_1}}
$$
until time $n$, the spheres are disjoint and there is no chance for a common neighbour. At  time $t$ such that $T \le t = o(t^+)$, the spheres overlap to a large extent and~(\ref{eq:ratio}) holds. However, for $\eps>0$ and $t$ such that $\eps t^+ \le t \le t^+$ only a nontrivial fraction of $S(v_\ell,t)$ is contained in $S(v_k,t)$. The above analysis still applies, but in this case instead of an asymptotic result,  we obtain the order result stated in the theorem. 

\bigskip

Finally, let us note that the number of common neighbours is a sum of independent random indicator variables with Bernouilli distribution. The concentration follows from the bound (\ref{lem:Chernoff}).
\end{proof} 

\subsection{Edge length distribution}

Finally, we give the proof of the theorem about the edge length distribution. Remember that a long edge is an edge such that its endpoints are at distance at least $r_\alpha$, where $r_\alpha$ is chosen so that a ball of radius $r_\alpha$ has volume $n^{-\alpha}$. As in the previous subsection, the proof distinguishes three cases, but now the three cases depend on whether the sphere of influence of a vertex has radius greater than $r_\alpha$ (allowing the vertex to receive long edges) at the beginning and the end of its life. 

First, we need to recall a few known results: the behaviour of $N_k=N_k(n)$, the number of vertices of in-degree $k=k(n)$ at time $n$, the number of edges $M=M(n)$ at time $n$, and the upper bound for the size of the influence regions. The following result was proven in~\cite{spa1}.

\begin{thm}[\cite{spa1}]\label{thm:spa1}
Suppose that $pA_1 < 1$. The following holds \aas\ for every $0 \le k \le (n/\log^8 n)^{(pA_1)/(4pA_1+2)}$.
$$
N_k = (1+o(1)) c_k n,
$$
where $c_0 = 1/(1+pA_2)$ and for $k \ge 1$, 
$$
c_k = \frac {p^k}{1+kpA_1+pA_2} \prod_{j=0}^{k-1} \frac {jA_1+A_2}{1+jpA_1+pA_2}.
$$
Moreover, \aas\ 
$$
M = (1+o(1)) \frac{pA_2}{1-pA_1} n.
$$
\end{thm}

Note that 
\begin{eqnarray*}
c_k &=& \frac{1}{pA_1} \cdot \frac { \prod_{j=0}^{k-1} \left(j + \frac{A_2}{A_1} \right) } { \prod_{j=0}^{k} \left(j + \frac{A_2}{A_1} +\frac{1}{pA_1} \right) } \\
&=& \frac{1}{pA_1} \cdot \frac { \Gamma \left(k + \frac{A_2}{A_1} \right) ~~/~~ \Gamma \left(\frac{A_2}{A_1} \right) } { \Gamma \left(k+1 + \frac{A_2}{A_1} +\frac{1}{pA_1} \right) ~~/~~ \Gamma \left(\frac{A_2}{A_1} +\frac{1}{pA_1} \right) } \ .
\end{eqnarray*}
Suppose now that $k=k(n)$ tends to infinity together with $n$. Using Stirling's asymptotic approximation of the Gamma function ($\Gamma(z) = (1+o(1)) \sqrt{2 \pi} z^{z-1/2} e^{-z}$) we can take $c_k$ to be:
$$
c_k = \frac{1}{pA_1} \cdot \frac {\Gamma \left(\frac{A_2}{A_1} +\frac{1}{pA_1} \right)} {\Gamma \left(\frac{A_2}{A_1} \right)} \ k^{-1-\frac{1}{pA_1}}\ ,
$$
and the following useful corollary is proved.

\begin{corollary}\label{cor:n_k}
Suppose that $pA_1 < 1$.  Let $\omega=\omega(n)$ be any function tending to infinity with $n$. The following holds \aas\ for every $\omega \le k \le (n/\log^8 n)^{(pA_1)/(4pA_1+2)}$.
$$
N_k = (1+o(1)) c k^{-1-\frac{1}{pA_1}} n,
$$
where  
\begin{equation}
\label{eqn:c}
c = \frac{1}{pA_1} \cdot \frac {\Gamma \left(\frac{A_2}{A_1} +\frac{1}{pA_1} \right)} {\Gamma \left(\frac{A_2}{A_1} \right)}.
\end{equation}
\end{corollary}

In~\cite{spa1}, it was proved that \aas\ for all vertices we have that $\deg^-(v_i,n) = O( (\log^2 n) (n/i)^{pA_1} )$, provided that $v_i$ was born at time $i$. Now, with Theorem~\ref{thm:deg} in hand, we get a stronger result, namely that \aas~for all $i\le t\leq n$
$$
\deg^-(v_i,t) = O \left( (\omega \log n) \left( \frac ti \right)^{pA_1} \right),
$$
where $\omega=\omega(n)$ is any function tending to infinity together with $n$. (Indeed, for a contradiction suppose that $\deg^-(v_i,t) \ge (2 \omega \log n) (t/i)^{pA_1}$ for some value of $t$. Theorem~\ref{thm:deg} implies that $\deg^-(v_i,i) = (2+o(1)) \omega \log n$ which is clearly a contradiction.) This implies the following result. 

\begin{thm} \label{thm:birth}
Suppose that $pA_1 < 1$, and $\omega$ is a function that goes to infinity together with $n$. The following holds \aas\ for every vertex born at time $i$.
$$
|S(v_i, t)| = O \left( \frac{\omega \log n} {i} \right) ,
$$
\end{thm}

The results given above are used in the proof of Theorem~\ref{thm:edge}, which we are now ready to give.

\begin{proof}[Proof of Theorem~\ref{thm:edge}]
Suppose first that $\alpha>1$. Since the sphere of influence of every vertex at every time of the process is (deterministically) at least $A_2/n \gg n^{-\alpha}$, ``long'' edges can occur at every step of the process. A vertex $v$ will receive a short edge precisely when the new vertex falls within a ball of radius $r_\alpha$ around $v$, and thus automatically falls within the sphere of influence of $v$, and then links to $v$. The probability that this happens equals $pn^{-\alpha}$. Thus, the expected number of short edges pointing to a vertex  born at time $i$  is $p n^{-\alpha} (n-i)$, and the total number of short edges is $(1+o(1)) p n^{2-\alpha} / 2 = o(n)$ and so is negligible compared to the total number of edges. We conclude that \aas\ almost all edges are long, and the result holds by Theorem~\ref{thm:spa1}.

Suppose now that $1- \frac{pA_1}{4pA_1+2} < \alpha < 1$. Let $e_v(\alpha)$ be the number of long edges pointing to $v$, that is:
$$
e_v(\alpha) = \Big| \big\{w\in N^-(v): d(v,w) \geq r_{\alpha}\big\} \Big|,
$$
where $N^-(v)$ is the in-neighbourhood of vertex $v$. 

For a vertex $v$ to receive an edge of length greater than $r_{\alpha}$ at time $t$, its region of influence must have radius at least $r_\alpha$, and thus have volume $|S(v,t)|\geq n^{-\alpha}$. Key to the proof is Theorem~\ref{thm:deg} and its conclusion that the regions of influence tend to be shrinking.  

Let $\omega=\omega(n)$ be any function increasing with $n$. First, we only consider vertices whose final degree is at least $\omega\log n$. This is enough to get a lower bound for the number of long edges. Later we will show that the contribution of the remaining edges is negligible. Consider a vertex $v$ with final degree $k=\deg^-(v,n)\geq \omega\log n$. It follows from Theorem~\ref{thm:deg} that \aas\ for every vertex $v$ of degree $k\geq \omega \log n$ at time $n$, 
\[
\deg^-(v,t)=(1+o(1)) k \left(\frac{t}{n}\right)^{pA_1}
\]
for all $t_k \le t \le n$, where
$$
t_k = n \left( \frac {\omega \log n}{k} \right)^{\frac {1}{pA_1}}.
$$
(Note that $\deg(v,t_k) = (1+o(1)) \omega \log n$.) Therefore, we may assume, without loss of generality, that for all $t_k \le t \le n$
$$
|S(v,t)| = (1+o(1)) A_1 k n^{-pA_1} t^{pA_1-1}.
$$

We distinguish  three possible classes of vertices, based on their final degree: vertices of high final degree can receive long edges from time $t_k$ until the end of the process, $t=n$ (Case~1); vertices with final degree in a mid-range can receive long edges from time $t_k$ until some time $t_k^*$, $t_k<t_k^*<n$ (Case~2); and vertices with small final degree can never receive long edges after time $t_k$ (Case~3). 

The cut-off values of the three cases are 
$$
k_{\min}=\left( \frac{n^{1-\alpha}}{A_1} \right)^{pA_1} \left( \omega \log n \right)^{1-pA_1} 
$$
and $k_{\max} = \frac {n^{1-\alpha}}{A_1}$. Consider a vertex $v$ of degree $k$.

\smallskip

{\bf Case 1.} Suppose that $k \ge k_{\max}$. Note that this implies that 
\[
|S(v,n)| = (1+o(1)) A_1 k / n \ge (1+o(1)) n^{-\alpha},
\]
so for any constant $\eps>0$, and for any time $t$ in the range $t_k \le t \le (1-\eps) n$,  the sphere of influence of $v$ has radius greater than $r_\alpha$. This implies that $v$ has an opportunity to receive long edges from time $t_k$ until the end of the process, or very close to it.

For $t_k\leq t \leq n$, the probability that $v$ receives a short edge (edge from a vertex within distance $r_\alpha$) equals $p \min\{ n^{-\alpha}, |S(v,t)| \} = (1+o(1)) pn^{-\alpha}$. Moreover, these events are independent. Thus, \wep\ the number of short edges is  
\[
(1+o(1))pn^{-\alpha}(n-t_k) = (1+o(1)) p n^{1-\alpha},
\]
where the last step uses the fact that $t_k=o(n)$ in this case. 

The degree of $v$ at time $t_k$ is $O(\omega \log n)$, so we have that \wep
\begin{eqnarray*}
e_v(\alpha) &=& \deg^-(v,n) - (1+o(1)) pn^{1-\alpha} + O(\omega \log n) = (1+o(1)) (k - pn^{1-\alpha}) \\
&\ge& (1-pA_1+o(1)) \frac{n^{1-\alpha}}{A_1}.
\end{eqnarray*}
Note that if $k \ge \omega n^{1-\alpha}$, then \wep\ almost all edges pointing to $v$ are long.

\smallskip

{\bf Case 2.}  Let $\eps>0$ be some (arbitrarily small) constant.  Suppose that $(1+\eps)k_{\min}\leq k\leq (1-\eps) k_{\max}$. The upper bound on $k$ implies that $|S(v,n)| \leq (1-\eps+o(1)) n^{-\alpha}$ so there is no chance for $v$ to receive long edges near the end of the process. On the other hand, it follows from the lower bound on $k$ that $|S(v,t_k)| \geq (1+\eps -o(1)) n^{-\alpha}$ so  if the new vertex at time $t_k$ falls within $S(v_t,k)$ , there is a positive probability that a long edge to $v$ is created.

Let 
$$
t^*_k = \left( A_1 k n^{\alpha-pA_1} \right)^{\frac{1}{1-pA_1}}.
$$
Note that $|S(v,t_k^*)|=(1+o(1))n^{-\alpha}$. Thus, the influence region of $v$ has radius greater than $r_\alpha$ from time $t_k$ to $(1-\delta) t^*_k$, and radius less than $r_\alpha$ from time $(1+\delta) t^*_k$ to $n$, for some small $\delta >0$. 

Thus, by a similar argument to the previous case, we obtain:
\begin{eqnarray*}
e_v(\alpha) &\ge& (1+o(1)) \sum_{t=t_k}^{(1-\delta)t_k^*} p \left( A_1 k n^{-pA_1} t^{pA_1-1} - n^{-\alpha} \right) \\
&=& (1-O(\delta)) \left( k n^{-pA_1} (t_k^*)^{pA_1} - p (t^*_k) n^{-\alpha} \right) \\
&=& (1-O(\delta)) A_1^{\frac {pA_1}{1-pA_1}} k^{\frac {1}{1-pA_1}} n^{\frac{-pA_1 (1-\alpha)}{1-pA_1}} \left( 1-pA_1 \right).
\end{eqnarray*}
Similarly, we get that $e_v \le (1+O(\delta)) A_1^{\frac {pA_1}{1-pA_1}} k^{\frac {1}{1-pA_1}} n^{\frac{-pA_1 (1-\alpha)}{1-pA_1}} \left( 1-pA_1 \right)$ and so
$$
e_v (\alpha)= (1+o(1)) A_1^{\frac {pA_1}{1-pA_1}} k^{\frac {1}{1-pA_1}} n^{\frac{-pA_1 (1-\alpha)}{1-pA_1}} \left( 1-pA_1 \right),
$$
by taking $\delta \to 0$.

\smallskip

{\bf Case 3.} Finally, suppose that $\omega \log n \le k \le (1-\eps)k_{\min}$ for some $\eps >0$. Since $|S(v,t_k)| \le (1-\eps +o(1)) n^{-\alpha}$, the influence region has radius smaller than $r_\alpha$ from time $t_k$ until the end of the process. Thus, for such vertices, all edges they receive in this time slot are short. Thus the only long edges $v$ can receive are those received before $t^*_k$, so $e_v(\alpha)= O(\omega \log n)$. Trivially, the same property holds for any vertex of degree smaller than $\omega \log n$. 

\smallskip

In order to obtain upper and lower bounds on the total number of long edges, we can use Theorem~\ref{thm:spa1} and its corollary (Corollary~\ref{cor:n_k}) to calculate the number of long edges pointing to vertices of final degree larger than $k_{\min}$ (Cases~1 and~2). Let $c$ be as defined in Equation (\ref{eqn:c}), and let $K=(n/\log^8 n)^{(pA_1)/(4pA_1+2)}$. By Corollary \ref{cor:n_k}, $K$ is the upper bound on the values of $k$ for which we have concentration for $N_k$. Note that the bounds on $\alpha$ imply that $k_{\max} \ll K$, and thus $\sum_{k\ge k_{\max}} k^{-\gamma}=(1+o(1))\sum_{k_{\max}\le k \leq K} k^{-\gamma}$ for all $\gamma >1$.

The number of long edges to vertices of the first type (Case~1) is \aas\ equal to
\begin{eqnarray*}
E_1 &=& (1+o(1)) \sum_{k \ge k_{\max} } N_k \left( k - pn^{1-\alpha} \right) \\ 
&= & (1+o(1)) \sum_{k_{\max} \le k\le K} \left( c k^{-1-\frac {1}{pA_1}} n \right) \left( k - pn^{1-\alpha} \right) \\ 
&=& (1+o(1)) \left( cn \sum_{k \ge k_{\max} } k^{-\frac{1}{pA_1}} - cpn^{2-\alpha} \sum_{k \ge k_{\max} } k^{-1-\frac{1}{pA_1}} \right) \\
&=& (1+o(1)) \left( cn \ \frac{(k_{\max})^{ \frac{pA_1-1}{pA_1}}}{\frac{1-pA_1}{pA_1}}  - cpn^{2-\alpha} \frac { (k_{\max})^{- \frac{1}{pA_1}} }{\frac {1}{pA_1}} \right) \\
&=& c \ \frac { pA_1^{\frac {1}{pA_1}}}{1-pA_1} n^{2-\frac {1}{pA_1} + \alpha \frac {1-pA_1}{pA_1}} \big( 1 - (pA_1)(1-pA_1) \big).
\end{eqnarray*}

The number of long edges to vertices of the second type (Case~2) is \aas\ equal to
\begin{eqnarray*}
E_2 &=& (1+o(1)) \sum_{k =k_{\min}}^{k_{\max}} \left( c k^{-1-\frac {1}{pA_1}} n \right) \left( A_1^{\frac {pA_1}{1-pA_1}} k^{\frac {1}{1-pA_1}} n^{\frac{-pA_1 (1-\alpha)}{1-pA_1}} \left( 1-pA_1 \right) \right) \\
&=& (1+o(1)) c \ A_1^{\frac{pA_1}{1-pA_1}} (1-pA_1) n^{1 - \frac {pA_1 (1-\alpha)}{1-pA_1}} \sum_{k =k_{\min}}^{k_{\max}} k^{-1 + \frac {2pA_1-1}{(1-pA_1)pA_1}}.
\end{eqnarray*}
(Technically, to get a lower bound of $E_2$ we should sum over $k_{\min} (1+\eps) \le k \le k_{\max} (1-\eps)$ and sum over $k_{\min} (1+\eps) \le k \le k_{\max}$ to get an upper bound. Since the error in this summation is $(1+O(\eps))$, the result holds by taking $\eps \to 0$.)

Since $1/2 < pA_1 < 1$, the exponent of $k$ in the summation is in the interval $(-1,0)$, and thus the behaviour of the summation is determined by its upper bound $k_{\max}$. This leads to 
\begin{eqnarray*}
E_2 &=& (1+o(1)) c \ A_1^{\frac{pA_1}{1-pA_1}} (1-pA_1) n^{1 - \frac {pA_1 (1-\alpha)}{1-pA_1}} 
\frac{ (k_{\max})^{\frac {2pA_1-1}{(1-pA_1)pA_1}}} {\frac {2pA_1-1}{(1-pA_1)pA_1}}\\
&=& (1+o(1)) c \ \frac { pA_1^{\frac {1}{pA_1}}}{1-pA_1} n^{2-\frac {1}{pA_1} + \alpha \frac {1-pA_1}{pA_1}} \cdot \frac {(1-pA_1)^3}{2pA_1-1} A_1^{\frac {1-2pA_1}{(1-pA_1)pA_1}}.
\end{eqnarray*}
Since $E_1$ and $E_2$ are of the same order, we can take $E_1+E_2$ as a lower bound for $e(\alpha)$.

In order to obtain an upper bound, we consider edges to vertices that are in Case~3, that is, those that have final degree at most $k_{\min}$. It follows from Theorem~\ref{thm:birth} that any vertex that is able to receive long edges directed to vertices with small final degree has to have a time of birth $i\leq i_{\max}=  \omega n^\alpha \log n$. There are obviously at most $i_{\max}$ of such vertices, and each of them has $O(\omega\log n)$ long edges. So the number of long edges that we did not count yet is at most:
\[
E_3= O \big( (\omega n^\alpha \log n) (\omega \log n) \big) = n^{\alpha+o(1)}.
\]
Since $E_3$ is of smaller order than $E_1+E_2$, the result follows.
\end{proof}

For the proof of Theorem \ref{thm:edge-exp}, we use a large part of the previous proof. 

\begin{proof}[Proof of Theorem \ref{thm:edge-exp}]
For this theorem, we consider the expected value of $e(\alpha)$. Thus, we can use the expected values of $N_k$, and do not need to consider the cut-off on the values of $k$ for which the values of $N_k$ are concentrated. In~\cite{spa1}, it was shown that 
\[
\E (N_k)= (1+o(1)) c k^{-1-\frac{1}{pA_1}} n, \text{ for all } k\geq k_{max}. 
\]

Suppose first that $1/2< pA_1< 1$ and $1- pA_1 < \alpha < 1$.
Consider the proof of Theorem \ref{thm:edge}. The three cases of this proof still hold as before; let $k_{\min}$ and $k_{\max}$ be as defined in this proof. As explained in this proof, concentration for the values of $N_k$ hold only up to degree $K=n^{pA_1/(4pA_1+2)+o(1)}$. This affects the computation of $E_1$.  However, in this proof we only consider the expected value of $e(\alpha)$, so by linearity of expectation, in the computation of $E_1$ we can use the expected values of the $N_k$.  This leads to the following expression for the expected number of long edges to vertices of the first type (Case~1) :
\begin{eqnarray}
\label{eqn:E(E1)}
\E (E_1) &=& (1+o(1)) \sum_{k \ge k_{\max} } \E (N_k) \left( k - pn^{1-\alpha} \right) \notag\\ 
&= & (1+o(1)) \sum_{k\ge k_{\max}  } \left( c k^{-1-\frac {1}{pA_1}} n \right) \left( k - pn^{1-\alpha} \right) \notag\\ 
&=& c \ \frac { pA_1^{\frac {1}{pA_1}}}{1-pA_1} n^{2-\frac {1}{pA_1} + \alpha \frac {1-pA_1}{pA_1}} \big( 1 - (pA_1)(1-pA_1) \big),
\end{eqnarray}
where $c$ is as defined in Equation (\ref{eqn:c}).

For the computation of $E_2$, we should note that we may not have concentration of $N_k$ for the values of $k$ close to  $k_{\max}$. However, we can a similar calculation to that used in  
the proof of Theorem~\ref{thm:edge}, using the expected values of the $N_k$, to obtain that $\E (E_2)=\Theta(n^{2-\frac {1}{pA_1} + \alpha \frac {1-pA_1}{pA_1}})$. 

The argument that $E_3$ is negligible compared to $\E (E_1)$ and $\E (E_2)$, as laid out in the proof of Theorem~\ref{thm:edge}, still holds here. Thus, we have that
\[
\E (e(\alpha)) = \Theta ( n^{2-\frac{1}{pA_1}+\alpha \frac{1-pA_1}{pA_1}})
\]
The result follows by taking the logarithm.

Next, consider the case where $1/2 < pA_1 <1$ and $\alpha < 1-pA_1$. It follows from Theorem \ref{thm:deg}, and it was also shown in \cite{spa1}, that \wep\ the maximum in-degree in a graph produced by the SPA model is at most $K_M=O( n^{pA_1}\log^4 n)$. Since $k_{\max}=n^{1-\alpha}/A_1 \gg n^{pA_1}$, \wep\ no vertices are in Case~1, so no vertices can receive long edges until the edge of the process.

For the vertices that are in Case 2, we can apply the same calculation as in the proof of Theorem~\ref{thm:edge}, while taking the expected values of the $N_k$ as in the previous case. Since \wep\ $K_M$ is an upper bound on the maximum degree, the expected number of vertices of degree greater than $K_M$ is at most $n \exp(-\Theta(\log^2 n))$. Hence, the expected number of long edges to such vertices, is at most $n^2 \exp(-\Theta(\log^2 n)) = o(1)$.  The expected number of long edges to vertices of the second type (Case~2) therefore is equal to
\begin{eqnarray*}
\E (E_2) &\leq & (1+o(1)) \sum_{k =k_{\min}}^{K_M} \left( c k^{-1-\frac {1}{pA_1}} n \right) \left( A_1^{\frac {pA_1}{1-pA_1}} k^{\frac {1}{1-pA_1}} n^{\frac{-pA_1 (1-\alpha)}{1-pA_1}} \left( 1-pA_1 \right) \right) \\
&=& (1+o(1)) c \ A_1^{\frac{pA_1}{1-pA_1}} (1-pA_1) n^{1 - \frac {pA_1 (1-\alpha)}{1-pA_1}} \sum_{k =k_{\min}}^{K_M} k^{-1 + \frac {2pA_1-1}{(1-pA_1)pA_1}}.
\end{eqnarray*}
For a lower bound on $\E (E_2)$, we should sum over $k_{\min} (1+\eps) \le k \le n^{pA_1}$.

Since $pA_1 > 1/2$, as before the exponent of $k$ in the summation  is determined by its upper bound $K_M=O(n^{pA_1}\log^4 n)$. This leads to 
\begin{eqnarray*}
\E (E_2) &\leq & (1+o(1)) c \ A_1^{\frac{pA_1}{1-pA_1}} (1-pA_1) n^{1 - \frac {pA_1 (1-\alpha)}{1-pA_1}} \frac{ (K_M)^{\frac {2pA_1-1}{(1-pA_1)pA_1}}} {\frac {2pA_1-1}{(1-pA_1)pA_1}}\\
&=& g(n) n^{\frac{pA_1\alpha}{1-pA_1}},
\end{eqnarray*}
for some function $g$ of order $g(n)=\Theta(K_M/n^{pA_1})=\Theta(\log^4 n)$. For the lower bound on $\E (E_2)$, we have the same summation, but with upper bound $n^{pA_1}$ instead of $K_M$, and thus 
$\E (E_2) = \Omega (n^{\frac{pA_1\alpha}{1-pA_1}})$.

Since $\log(g(n))= o(\log n)$, we can combine lower and upper bound to see that
\[
\frac{\log \E (E_2)}{\log n} =  \frac{pA_1\alpha}{1-pA_1}+ o(1).
\]

Finally, consider the vertices in Case 3. Here, the exact same argument as given in the proof of Theorem~\ref{thm:edge}  can be used to show that
\[
\E (E_3)= O \big( (\omega n^\alpha \log n) (\omega \log n) \big) = n^{\alpha+o(1)}.
\]
Since this is of smaller order than $\E (E_2)$, the result follows.

Finally, consider the case where $pA_1\leq 1/2$. For $\alpha \in (1-pA_1,1)$ we have the exact same expression for $\E(E_1)$ as for the case where $pA_1>1/2$, as given in Equation (\ref{eqn:E(E1)}). Thus
\[ 
\E (E_1)= \Theta (n^{2-\frac {1}{pA_1} + \alpha \frac {1-pA_1}{pA_1}}) = o(n^{\alpha}),
\]
where the last step follows since
\[
2-\frac {1}{pA_1} + \alpha \frac {1-pA_1}{pA_1} =1-(1-\alpha)\left(\frac{1-pA_1}{pA_1}\right)<\alpha.
\]

For $\alpha \in (0,1-pA_1)$ we have that \wep\ $E_1=0$, so $\E (E_1)=\exp(-\Theta(\log^2 n ))$.

For $E_2$, we have the same sum as before: let $K^*=k_{\max}$ if $\alpha\in (1-pA_1,1)$, and $K^*$, an (almost sure) upper bound on the maximum degree, otherwise.  Then
\begin{eqnarray*}
\E (E_2) &= & (1+o(1)) \sum_{k =k_{\min}}^{K^*} \left( c k^{-1-\frac {1}{pA_1}} n \right) \left( A_1^{\frac {pA_1}{1-pA_1}} k^{\frac {1}{1-pA_1}} n^{\frac{-pA_1 (1-\alpha)}{1-pA_1}} \left( 1-pA_1 \right) \right) \\
&=& (1+o(1)) c \ A_1^{\frac{pA_1}{1-pA_1}} (1-pA_1) n^{1 - \frac {pA_1 (1-\alpha)}{1-pA_1}} \sum_{k =k_{\min}}^{K^*} k^{-1 + \frac {2pA_1-1}{(1-pA_1)pA_1}}.
\end{eqnarray*}

Since $pA_1 < 1/2$, the exponent of $k$ in the summation  in this case is determined by its lower bound 
\[
k_{\min}=\left( \frac{n^{1-\alpha}}{A_1} \right)^{pA_1} \left( \omega \log n \right)^{1-pA_1}.
\]
This leads to 
$$
\E (E_2) \leq  (1+o(1)) c \ A_1^{\frac{pA_1}{1-pA_1}} (1-pA_1) n^{1 - \frac {pA_1 (1-\alpha)}{1-pA_1}} \frac{ (k_{\min})^{\frac {2pA_1-1}{(1-pA_1)pA_1}}} {\frac {2pA_1-1}{(1-pA_1)pA_1}} = o(n^{\alpha}),
$$
where the last step follows since the exponent of $(\omega \log n)$ in $(k_{\min})^{\frac {2pA_1-1}{(1-pA_1)pA_1}} $ equals
\[
(1-pA_1) \frac {(2pA_1-1)}{(1-pA_1)pA_1} < 0.
\]

Finally, the same estimate as before can be used to show that $E_3\leq n^{\alpha +o(1)}$, and thus $\log(\E (e(\alpha))/\log n \le \alpha +o(1)$.

For the lower bound, note that all volumes of influence up to time $T=(A_1/2)n^{\alpha}$ have (deterministically) volume at least $2n^{\alpha}$. Therefore, a positive fraction of all edges generated until time $T$ are long, and so \aas\ $\Omega(n^{\alpha})$ is a lower bound for the number of long edges and the theorem is finished.
\end{proof}

\end{document}